\documentclass[table]{ccjnl}
\usepackage{lipsum,amsmath}
\usepackage{cuted}
\usepackage{bm}
\usepackage{color}
\usepackage{comment}
\usepackage{graphicx}
\usepackage{ulem}

\graphicspath{{figures/}}

\title{Sum-Rate Maximization in Active RIS-Assisted Multi-Antenna WPCN}
\author{Jie Jiang, ~Bin Lyu, ~Pengcheng Chen, ~and ~Zhen Yang}
\receiveddate{XX}
\reviseddate{XX}
\Editor{XX}

\address[ ]{School of Communications and Information Engineering, Nanjing University of Posts and Telecommunications, Nanjing 210003, China}

\begin{document}

\maketitle

\begin{abstract}
  In this paper, we propose an active reconfigurable intelligent surface (RIS) enabled hybrid relaying scheme for a multi-antenna wireless powered communication network (WPCN), where the active RIS is employed to assist both wireless energy transfer (WET) from the power station (PS) to energy-constrained users and  wireless information transmission (WIT) from users to the receiving station (RS). For further performance enhancement, we propose to employ both transmit beamforming at the PS and receive beamforming at the RS. We formulate a sum-rate maximization problem by jointly optimizing the RIS phase shifts and amplitude reflection coefficients for both the WET and the WIT, transmit and receive beamforming vectors, and network resource allocation. To solve this non-convex problem, we propose an efficient alternating optimization algorithm with  linear minimum mean squared error criterion, semi-definite relaxation (SDR) and successive convex approximation techniques. Specifically, the tightness of applying the SDR is proved. Simulation results demonstrate that our proposed scheme with 10 reflecting elements (REs) and 4 antennas can achieve 17.78\% and 415.48\% performance gains compared to the single-antenna scheme with 10 REs and passive RIS scheme with 100 REs, respectively. 
\end{abstract}

\section{Introduction}
With the development of the Internet-of-Things (IoT), an intelligent society with ubiquitous interconnections and deep coverage will be truly realized. However, wireless devices (WDs) in IoT networks are generally energy-constrained and suffer from limited lifetime \cite{somov2015Powering}, which fundamentally limits the performance of communication networks. Traditional ways of changing  or recharging the batteries manually are impossible and unacceptable, especially when the number of WDs is numerous. Therefore, how to tackle this  issue  is a critical problem in the widespread development of IoT. Wireless powered communication has been proposed as a prospective technology for enhancing the energy sustainability of WDs, which can be classified into two directions based on application scenarios \cite{varshney2008Transporting,zhong2014Wireless,wu2021Intelligent}. The first one focuses on investigating simultaneous wireless information and power transfer (SWIPT), where the base station (BS) simultaneously transfers energy and information signals to energy receivers and information receivers via the common radio frequency (RF) signals in the downlink (DL), resulting in a pivotal tradeoff between the achievable rate and harvested energy \cite{zhang2013MIMO}. In contrast to the SWIPT, wireless powered communication network (WPCN) has been proposed as a novel type of wireless network diagram to improve the lifetime of WDs and enhance the deployment flexibility of IoT. In a WPCN, energy-constrained WDs first harvest energy in the DL and then use the harvested energy to transmit independent information in the uplink (UL) based on the widely used harvest-then-transmit (HTT) protocol \cite{ju2013Throughput}. 

WPCN has been widely investigated in the literature \cite{ju2013Throughput,ju2014user,ju2014Optimal,kim2016Sum}, which promotes the development of WPCN. However, WPCN generally suffers from the  “doubly near-far” phenomenon if the power station (PS) and the receiving station (RS) are co-located at the hybrid access point (HAP) \cite{ju2013Throughput,ju2014Optimal}. Specifically, a WD located far away from the HAP harvesting less energy in the DL has to transmit information with more power in the UL, which results in an unfair time and resource allocation among the WDs. To deal with the issue, a promising way is to deploy the PS and RS separately \cite{wu2015Energyefficient,wu2015Wirelessa}.  In \cite{wu2015Energyefficient}, multiple users harvest energy from a dedicated PS and then communicate with an information RS following the HTT protocol. In this scenario, a user physically close to the PS is naturally far away from the RS, and vice versa. Considering a similar scenario, a user-centric energy-efficient (EE) problem in WPCN is investigated in \cite{wu2015Wirelessa}. However, the performance  of WPCN is still limited due to the low efficiencies caused by the severe path-loss, which seriously affects its practical applications.

Recently, reconfigurable intelligent surface (RIS), with the unprecedented ability to reshape the wireless transmission environment, has drawn widespread attentions from academia and industry \cite{wu2019Intelligent,direnzo2020Smart,kundu2020RISAssisted,zou2020Wireless,zhang2021Joint,liang2021Reconfigurablea,yu2021Smart}. RIS is comprised of a large number of programmable reflecting elements (REs), which can alter the phase shifts and amplitudes of incident signals. As such, RIS can adaptively modify the impinging radio waves towards the appropriate direction \cite{direnzo2020Smart}. According to the reflection patterns, RIS can be classified as passive RIS and active RIS. Without the property of power amplification, the independent diffusive scatterer-based (IDS)  model accounts for the basic properties of passive RIS, which has been widely adopted in RIS-assisted wireless communications \cite{yu2021Smart}. The passive RIS is only equipped with the phase-shift controller, while the active RIS contains both the phase-shift controller and the active reflection-type amplifier. Hence, the active RIS can alter both the phase shifts and amplitudes of the incident signals. It is worth noting that the active RIS with its novel hardware structure and signal model has been proposed in \cite{zhang2021Active,long2021Active}.  Different from the full-duplex amplify-and-forward (FD-AF) relay that requires power-consuming RF chains, the active RIS can directly reflect and amplify the incident signals in the EM level in a FD manner without reception. In this way,  the active RIS exhibits promising qualities, such as a low power consumption, light weight, conformal geometry and high flexibility for practical deployment.

Currently, the passive RIS has been widely applied in WPCNs for performance enhancement \cite{lyu2021Optimized,zheng2021Joint,hua2022PowerEfficient,xu2021RISenhanced}.  In \cite{lyu2021Optimized}, the passive RIS is employed between the HAP and users to improve both DL WET and UL WIT efficiencies in single-input-single-output (SISO) WPCN. To achieve further performance improvement, the multi-antenna technique is employed in \cite{zheng2021Joint}, where  the HAP with multi-antenna transmits energy signals to users in the DL and receives information from users in the UL by employing  transmit bemforming and receive beamforming, respectively. In \cite{hua2022PowerEfficient}, the fully  dynamic RIS beamforming scheme is proposed for WPCN, for which the phase shift vectors are independently designed over different time slots. However, in the above works \cite{lyu2021Optimized,zheng2021Joint,hua2022PowerEfficient}, the PS and  RS are also co-located at the HAP, which results in performance unfair among users. To address this issue, the authors in \cite{xu2021RISenhanced} consider the scenario where the PS and RS are separately deployed, for which the locations of RIS and users can be carefully considered to achieve a fair performance among users. However, as mentioned above, the passive RIS can only reflect incident signals without amplification, which leads to the limited performance enhancement due to the double-fading effect suffered by the reflecting links. Thus, the energy-constrained users still need to consume much time for harvesting energy in the DL and  have  less time for information transmission.

Inspired by the amplification characteristic of the active RIS, the active RIS is confirmed to be superior to the passive RIS in terms of performance enhancement, and thus has been considered a promising technique for IoT networks  \cite{zhang2021Active,long2021Active,you2021Wireless,dong2021Active,zargari2022Multiuser,gao2022Beamforming,zeng2022Throughput}. The authors in \cite{zhang2021Active} compare the capacity improvement achieved by the active RIS to the passive RIS, which demonstrats that the active RIS can fundamentally mitigate the double-fading effect. In \cite{long2021Active}, an active RIS is applied in single input multiple output (SIMO) systems, for which the joint optimization of phase shifts matrix and receive beamforming is considered to obtain the maximum achievable rate. In \cite{you2021Wireless}, the placement of the active RIS  is optimized to enhance SISO  systems' performance. The authors in \cite{dong2021Active} propose to use the active RIS to  achieve secure transmission, which not only establishes the reliable link from the transmitter to the receiver but also prevents  the confidential information intercepted by the eavesdropper. The active RIS-aided multiuser MISO PS-SWIPT is studied in \cite{zargari2022Multiuser} to minimize the base station transmit power, which shows significant improvements compared to the passive RIS-aided system.  Similarly, in \cite{gao2022Beamforming}, an active RIS is employed to assist SWIPT to boost the efficiency of both WET and WIT, while the conclusions and approaches are inapplicable to the WPCN system because of thire different system models.

Although the active RIS has received a lot of interests for wireless communication networks, the applications of active RIS in WPCN is still at the very early stage and has not been well studied in the literature. To the best of our knowledge, there exists only one paper investigating the usage of active RIS in WPCN \cite{zeng2022Throughput}. Specifically, the authors in \cite{zeng2022Throughput} investigate the weight sum-rate maximization problem in the active RIS assisted single-antenna WPCN, where the PS and RS is co-located at the HAP. As a result, similar to \cite{lyu2021Optimized,zheng2021Joint,hua2022PowerEfficient}, a part of WDs in  \cite{zeng2022Throughput} still suffer from the “doubly near-far” phenomenon, which is not suitable for practical applications with high requirement of performance fairness. Moreover, the authors consider a certain simplified communication scenario where both the HAP and the WDs have single antenna each. The transmit beamforming and receive beamforming cannot be exploited, which is also a key technology for performance enhancement.

Motivated by the observations above, we propose an active RIS enabled hybrid relaying scheme for the multi-antenna WPCN, where the active RIS is employed to facilitate both the WET from the PS to energy-constrained users and the WIT from users to the RS, which is shown in Figure 1. Compared with the existing works which used the passive RIS in WPCN \cite{lyu2021Optimized,zheng2021Joint,hua2022PowerEfficient,xu2021RISenhanced}, our proposed active RIS scheme can amplify the energy signals and information signals to achieve a satisfying system performance. Different from the single-antenna scenario considered in \cite{zeng2022Throughput}, we propose to employ multi-antenna at both the PS and RS, which can construct the transmit beamforming and receive beamforming for further performance enhancement. Specifically, the transmit beamforming at the PS can be used to enhance the WPT efficiency from the PS to users. In the meanwhile, the receive beamforming at the RS can be used to exploit the antenna gain and eliminate the noise caused by the active RIS. In the considered system setup, we aim to maximize the sum-rate problem by jointly optimizing the transmit beamforming at the PS, the receive beamforming at the RS, the reflecting coefficients at the RIS, and network resource allocation. It should be noted that compared to  \cite{lyu2021Optimized,zheng2021Joint,hua2022PowerEfficient,xu2021RISenhanced,zeng2022Throughput}, our formulated problem is much more challenging to solve and the proposed algorithms in \cite{lyu2021Optimized,zheng2021Joint,hua2022PowerEfficient,xu2021RISenhanced,zeng2022Throughput} cannot be used to solve our formulated problem. Thus, we propose an efficient algorithm to solve the formulated problem. The main contributions of this paper are summarized as follows:
\begin{itemize}
\item{We propose an active RIS assisted multi-antenna WPCN for performance enhancement, where the active RIS is served as a hybrid relay to achieve two purposes, i.e., the first one is to assist the WET from the PS to users, and the second one is to aid the WIT from users to the RS. To further improve system performance, both transmit beamforming and receive beamforming techniques are respectively considered at the PS and RS.}

\item{We investigate the sum-rate maximization problem by jointly optimizing transmit beamforming vector at the PS, receive beamforming vectors at the RS, phase shifts and amplitude reflection coefficients at the RIS for both the WET and the WIT, and network resource allocation. To deal with the non-convexity of the formulated problem, we propose an efficient alternating optimization (AO) algorithm. Specifically, the original problem can be divided into four sub-problems, which are solved sequentially in an alternating manner until convergence is achieved.}

\item{For designing the receive beamforming, we apply the linear minimum mean squared error (MMSE) criterion and obtain the closed-form expression. For the optimization of the transmit beamforming and RIS reflecting coefficient matrices for the WET, the semidefinite relaxation (SDR) technique is adopted and the tightness of applying the SDR is proved. For the optimization of RIS reflecting coefficients for the WIT, we obtain the optimal phase shifts in a closed-form and propose a successive convex approximation (SCA) algorithm to determine the optimal amplitude reflection coefficients. In addition, the convergence of proposed problem is analyzed and confirmed via numerical simulations.}

\item{Finally, numerical results are provided to evaluate the performance of proposed scheme, which indicates that compared to the single-antenna scheme with 10 REs and the passive-RIS scheme with 100 REs, the proposed scheme with 4 antennas and 10 REs can achieve 17.78\% and 415.48\% sum-rate gain, respectively.}

\end{itemize}

The rest of this paper is organized as follows. Section II describes the system model of the active RIS-assisted multi-antenna WPCN. The sum-rate maximization problem is formulated in Section III and solved in Section IV, respectively. In Section V, performance is evaluated by numerical results. Finally, this paper is concluded in Section VI.

Notations: In this  paper, vectors and matrices are denoted by boldface lowercase and uppercase letters, respectively. The operators $ (\,\cdot\, )^T$, $ (\,\cdot\, )^H$, $ \left|\, \cdot \,\right| $ and $\left\|\,\cdot \,\right\| $ denote the transpose, conjugate transpose, absolute value and the Euclidean norm, respectively. $\text{Tr}(\,\cdot\, )$ and $\text{rank}(\,\cdot\, )$ denote the trace and rank of a matrix, respectively. $\bm{X}\succeq 0$ represents that $\bm{X}$ is a positive semidefinite matrix. $\mathbb{E} [\,\cdot\, ]$ stands for the statistical expectation. $\bm{I}_N$ denotes the $N$-dimensional identity matrix. $\bm{0}$ denotes the zero matrix/vector with appropriate size. $\mathbb{C} ^{N\times M}$ denotes the set of all $N\times M$ complex-valued matrices. $\mathbb{H}^M$ denotes the set of all $M\times M$ Hermitian matrices. $\mathbb{R}^{N\times 1}$ represents the set of all $N\times 1$ real-valued vectors. ${\mathcal{C} \mathcal{N} }\left(\mu ,\sigma^2\right)$ denotes the distribution of a circularly symmetric complex Gaussian random vector with mean $\mu $ and variance $\sigma^2$. $\text{arg}(\,\cdot\, )$ denotes the phase extraction operation.  $\text{diag}( \bm{x} )$ denotes a diagonal matrix whose diagonal elements are from vector $\bm{x} $. $\mathrm{Im}(\,\cdot\, )$ and $\mathrm{Re}(\,\cdot\, )$ respectively  denotes the imaginary part and real part of a complex number.

\section{System Model }
\label{sysmod}
\begin{figure}[t]
  \centering
  \includegraphics[width=0.45\textwidth]{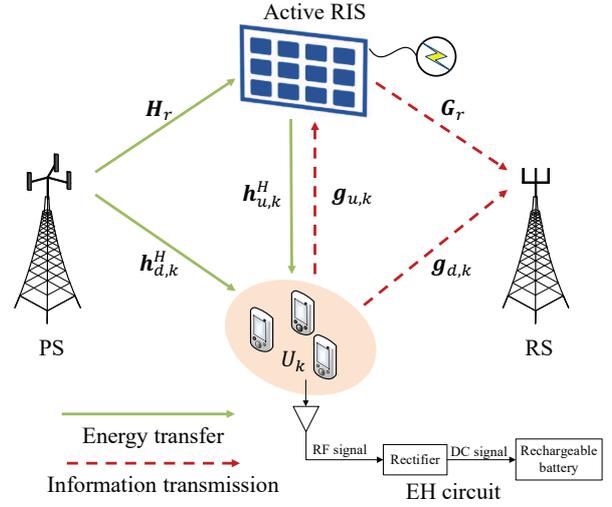}
  \caption{\label{fig:model} System model for an active RIS-assisted multi-antenna WPCN.
  }
\end{figure}

As shown in Figure \ref{fig:model}, we consider an active RIS-aided multi-antenna WPCN, which consists of a PS with $M$ antennas, a RS with $L$ antennas, an active RIS, and $K$ energy-constrained users each with single antenna. We assume that each user is equipped with an energy harvesting (EH) circuit for harvesting energy, where a rectifier is used to convert the received RF signals to direct current (DC) signals. Then, the net energy harvested from the DC signals can be stored in the rechargeable battery. 
The active RIS consists of $N$ active REs, which can steer the reflected signals in a specific direction and also amplify them by the active loads (negative resistance) \cite{long2021Active}. In contrast to the passive RE, each active RE is equipped with an additionally integrated active reflection-type amplifier supported by a power supply. By appropriately setting the effective resistance, it is reasonable to assume that the reflection amplitude and phase of each element is independently \cite{zhang2021Active,long2021Active,you2021Wireless,dong2021Active,zargari2022Multiuser,gao2022Beamforming,zeng2022Throughput}. To power the operations of active RIS and users, the PS is equipped with a stable energy source. In addition, the PS has the capability for performing computational tasks \cite{lyu2021Optimized}. In particular, the users first harvest energy from the RF signals transmitted by the PS and then use the harvested energy to deliver information to the RS. To allievate the severe path-loss suffered by the reflecting links, the active RIS is employed to improve the WET efficiency from the PS to users and the WIT efficiency from the users to the RS.

The channels are assumed to follow a quasi-static flat-fading model. That is, all channel coefficients are constant throughout each  transmission block but vary from block to block \cite{wu2020Weighted}. The downlink baseband equivalent channels of PS-to-RIS, RIS-to-$U_k$, and PS-to-$U_k$ links are denoted by $\bm{H}_r\in \mathbb{C} ^{N\times M}$, $\bm{h}_{u,k}^H\in \mathbb{C}^{1\times N}$, and $\bm{h}_{d,k}^H\in \mathbb{C}^{1\times M}$, respectively, where $U_k$ denotes the $k$-th user. Similarly, the uplink baseband equivalent channels of $U_k$-to-RIS, RIS-to-RS, and $U_k$-to-RS links are respectively denoted by $\bm{g}_{u,k}\in \mathbb{C}^{N\times 1}$, $\bm{G}_r\in \mathbb{C}^{L\times N}$, and  $\bm{g}_{d,k} \in \mathbb{C}^{L\times 1}$, respectively. Since there have been many efficient channel estimation techniques proposed for RIS systems \cite{hu2021robust,zheng2019Intelligent,wu2021Intelligenta,wang2020Channel}, we assume the perfect channel state information can be available in advance, which is a common assumption  considered in \cite{lyu2021Optimized,zheng2021Joint,zheng2020Intelligent} and a prerequisite for investigating the upper-bound of system performance. It should be noted that the channel estimation error is generally inevitable \cite{hu2021robust}. However, the effect of channel estimation error on system performance degradation is out the scope of this paper.

According to the HTT protocol \cite{ju2013Throughput}, the normalized transmission block of interest is divided into $K+1$ time slots. The first time slot with duration of $\tau_0 \in [0,1]$ is a dedicated slot for WET, in which all users harvest energy from the PS with the assistance of active RIS. The remaining $K$ time slots denoted by $\bm{\tau} = [\tau_1, \ldots, \tau_K]$, are used for UL WIT via the time division multiple access (TDMA) scheme. Specifically, during $\tau_k$, $k=1,\cdots,K$, $U_k$  delivers its information to the RS. Without loss of generality, the whole operation time period is set to a normalized transmission block. The network time scheduling constraint is thus given by
\begin{equation}
  \tau_0+\sum_{k=1}^{K}\tau_k \leq  1,~\forall k.\label{time}
\end{equation}

\subsection{Wireless Energy Transfer Phase}

In the WET phase, the PS transmits energy signals to all users with the assistance of the active RIS. Denote the transmitted signal as $\bm{s}=\bm{w}_0s$, where $s$ is the pseudo-random baseband signal transmitted by the PS, and $\bm{w}_0 \in\mathbb{C}^{M\times1}$  is the transmit beamforming vector. The energy constraint at the PS is expressed as 
\begin{equation}
  \mathbb{E} \left[ |\bm{s}|^2\right]=\text{Tr}\left(\bm{w}_0\bm{w}_0^H\right)\le P_0 , \label{W0_WET}
\end{equation}
where $P_0$ denotes the maximum transmit power at the PS. 

The reflecting coefficient matrix of the active RIS in the WET phase is  denoted by $\bm{\Phi}_0=\text{diag}\left\{ \phi_{0,1},\ldots, \phi_{0,N} \right\} \in \mathbb{C}^{N \times N}$ with $\phi_{0,n} =a_{0,n}e^{j\theta_{0,n}} $, $n=\,1,\ldots,\,N$, where $a_{0,n}$ and $\theta_{0,n}$ represent the amplitude reflection coefficient and phase shift of the $n$-th RE, respectively. Without loss of generality, we suppose each active RE has the following constraints
\begin{equation}
  a_{0,n} \le {a_{n,max}},\,\, 0 \le {\theta_{0,n}} \le 2\pi, ~\forall n,\label{RE_WET}
\end{equation}
where $a_{n,max}$ is the maximum amplitude reflection coefficient of $n$-th RE. It is worth noting that $ a_{n,max}$ can be greater than 1 \cite{you2021Wireless}, which is a main characteristic distinguishing the active RIS from the passive RIS since the active load can amplify the reflected signals.

The received signal at $U_k$ during $\tau_0$ is given by 
\begin{equation}
    {y}_{u,k}=\underbrace{\bm{h}_{d,k}^H \bm{s}}_{\text{direct link}}+ \underbrace{\bm{h}_{u,k}^H \bm{\Phi}_0 \left(\bm{H}_r \bm{s}+ \bm{n}_v\right)}_{\text{RIS-aided link}}+ {n}_{u,k},~\forall k,\label{ReceivedSignal}
\end{equation}
where ${n}_{u,k} \in \mathbb{C}$ and $\bm{n}_v \in\mathbb{C}^{N\times1} $ represent the additive white Gaussian noise (AWGN) at $U_k$ and the RIS, respectively. Without loss of generality, we assume  ${n}_{u,k}\sim\mathcal{CN}\left(0,\sigma_{u,k}^2 \right)$ and $\bm{n}_v\sim \mathcal{CN}\left(0,\sigma_v^2 \bm{I}_N\right)$. Denote the equivalent downlink channel as $\bm{h}_k^H = \bm{h}_{u,k}^H \bm{\Phi}_0 \bm{H}_r+ \bm{h}_{d,k}^H \in \mathbb{C}^{1\times M}$ and  \eqref{ReceivedSignal} can be rewritten as 
\begin{equation}
  {y}_{u,k} = \bm{h}_k^H \bm{w}_0 s + \bm{h}_{u,k}^H \bm{\Phi}_0 \bm{n}_v+ {n}_{u,k}, ~\forall k.
\end{equation}
Due to the fact that the active RIS not only amplifies the desired signal, i.e., $\bm{s}$, but also amplifies the input noise, i.e., $\bm{n}_v$, it is reasonable to consider the second term of  \eqref{ReceivedSignal} for computing the amount of harvested energy accurately \cite{long2021Active}. 
However, the noise at $U_k$ is generally quite small and can be negligible. Accordingly, the harvested energy by $U_k$, denoted by $E_k$, is given by 
\begin{equation}
    E_k =  \beta {\left|{\bm{h}_k^H}\bm{w}_0 \right|^2} \tau_0 + \beta{{\left\|{\bm{h}_{u,k}^H}\bm{\Phi}_0 \right\|^2}{\sigma_v^2}} \tau_0, ~\forall k,
\end{equation}
where  $\beta\in\left(0,1\right]$ denotes the energy conversion efficiency of each user. It is practical for us to consider the linear EH model here, which is also a common assumption in the literature \cite{ju2013Throughput,ju2014Optimal,wu2015Wirelessa,wu2015Energyefficient,zheng2021Joint,zeng2022Throughput}.

It is worth noting that the active RIS can allocate the available reflecting power to amplify the incident signals with active loads \cite{long2021Active}. In the DL WET phase, the amplification power of $\bm{s}$ and $\bm{n}_v$ is limited by the RIS power budget, which is shown by the following constraint
\begin{equation}
  {P_0} \left\| \bm{\Phi}_0 \bm{H}_r \right\|^2 + {\sigma_v^2}\left\|\bm{\Phi}_0 \bm{I}_N\right\|^2\le P_r,\label{Pr_WET}
\end{equation}
where $P_r$ is the maximum reflecting power for amplification at the active RIS and substantially lower than that of an active RF amplifier \cite{you2021Wireless}.

\subsection{Wireless Information Transmission Phase}
In the WIT phase, the users utilize the harvested energy to transmit information to the RS via a TDMA manner. Let $f_k$ denotes the information-carrying signal of $U_k$  with unit power and then the transmit signal  during $\tau_k$ is denoted by $x_k=\sqrt{p_k}f_k$, where $p_k$ is the transmit power at $U_k$. We assume that all the harvested energy at $U_k$ in the WET phase is used for delivering its own information. Let $\bm{p} = [p_1,\ldots,p_K] \in \mathbb{R} ^{1 \times K}$, which satisfies
\begin{equation}
  p_k\tau_k\le E_k, ~\forall k.\label{E_k1}
\end{equation}

Similarly, the reflecting coefficient matrix at the active RIS for the WIT during $\tau_k$ is denoted by  $\bm{\Phi}_k= \text{diag}\left\{ \phi_{k,1},\ldots, \phi_{k,N} \right\} \in \mathbb{C}^{N \times N}$, where $\phi_{k,n} =a_{k,n}e^{j\theta_{k,n}} $, $a_{k,n}$ and $\theta_{k,n}$ have the following constraints
\begin{equation}
  a_{k,n} \le {a_{n,max}},\,\, 0 \le {\theta_{k,n}} \le 2\pi, ~\forall k, ~\forall n.\label{RE_WIT}
\end{equation}

The received signal at the RS from $U_k$ with the assistance of active RIS during $\tau_k$ is written as
\begin{equation}
  \bm{y}_{r,k}=\underbrace{\bm{g}_{d,k}^H x_k}_{\text{direct link}}+ \underbrace{\bm{G}_r \bm{\Phi}_k \left(\bm{g}_{u,k} x_k+ \bm{n}_v\right)}_{\text{RIS-aided link}}+ \bm{n}_r,~\forall k,\label{y_r}
\end{equation}
where  $\bm{n}_r\in\mathbb{C}^{L\times1} $ represents the AWGN at the RS and satisfies $\bm{n}_r\sim\mathcal{CN}\left(0,\sigma_r^2 \bm{I}_N \right)$. 

In the UL WIT phase, we also have the amplification power constraint at the active RIS as follows
\begin{equation}
  {p_k} \left\|\bm{\Phi}_k {\bm{g}_{u,k}}\right\|^2 + {\sigma_v^2}\left\|\bm{\Phi}_k \bm{I}_N\right\|^2\le P_r,~\forall k. \label{Pr_WIT}
\end{equation}

Denote the receive beamforming vector at the RS during $\tau_k$ as  $\bm{w}_k\in \mathbb{C}^{L\times 1}$, which can be used to extract the desired signal and suppress interference and noises. Let the equivalent uplink channel be $\bm{g}_k = \bm{g}_{d,k} + \bm{G}_r \bm{\Phi}_k \bm{g}_{u,k},~ \forall k$. The estimated signal at the RS during $\tau_k$ is expressed as 
\begin{equation}
  \begin{aligned}
    \bm{u}_k = &\bm{w}_k^H \bm{y}_{r,k}\\ 
    =&\bm{w}_k^H \bm{g}_k x_k + \bm{w}_k^H \bm{G}_r \bm{\Phi}_k \bm{n}_v + \bm{w}_k^H \bm{n}_r, ~\forall k.
  \end{aligned}
\end{equation}

Then, the signal-noise-ratio (SNR) at the RS during $\tau_k$ is written as
\begin{equation}
\label{SNR}
  \gamma_k= \frac{p_k\left|{\bm{w}_k^H}\left(\bm{G}_r \bm{\Phi}_k {\bm{g}_{u,k}} +{\bm{g}_{d,k}} \right)\right|^2}{\left\| \bm{w}_k^H \bm{G}_r \bm{\Phi}_k \right\|^2 {\sigma_v^2} + {\left\|{\bm{w}_k}\right\|^2\sigma_{r}^2}},~\forall k.
\end{equation}

Denote the achievable rate from $U_k$ to the RS as $R_k$, which  is  formulated as
\begin{equation}
  R_k = \tau_k \log \left(1+\gamma_k\right),~\forall k. 
\end{equation}

\section{PROBLEM FORMULATION}
In this section, we formulate the system sum-rate maximization problem by jointly optimizing the reflecting coefficients for both the WET and the WIT at the active RIS, the transmit beamforming at the PS, the receive beamforming at the RS, the transmit power at each user, and the network time scheduling. The optimization problem is formulated as 

\begin{align}
    (\mbox{\textbf{P1}}) &\max_{\substack{ \tau_0,\bm{\tau}, \bm{\Phi}_0, \bm{\Phi}_k, \bm{w}_0, \bm{p},{\bm{w}_k}}} \sum_{k=1}^K R_k \label{pro_obj}\\
    s.t. \,\,\,\,\,\,&\mbox{\eqref{time}, \eqref{W0_WET}, \eqref{RE_WET}, \eqref{Pr_WET}, \eqref{E_k1}, \eqref{RE_WIT} and \eqref{Pr_WIT},}\nonumber\\
    & {\tau_0} \ge 0,\, {\tau_k}\ge 0,  {p_k} \ge 0,\, ~\forall k. \label{positive}
\end{align}
where \eqref{positive} indicates that the time and power variables are all nonnegative.

One can observe that the objective function in \eqref{pro_obj} is a non-concave function due to the coupled of variables. In addition, there exists several non-convex constraints, i.e., \eqref{W0_WET}, \eqref{Pr_WET}, \eqref{E_k1} and \eqref{Pr_WIT}. It is thus challenging to solve \textbf{P1} directly by  standard optimization techniques.  In the next section, we propose an  AO algorithm to solve it efficiently.

\section{ALTERNATING OPTIMIZATION SOLUTION}

In this section, an efficient AO algorithm is proposed to solve \textbf{P1}. In particular, we  decompose \textbf{P1} into several subproblems and iteratively solve them in an alternating manner. To show the procedure of AO algorithm, we summarize a flow chart in Figure 2. Specifically, the variables are partitioned into four blocks, $\{\bm{w}_k \}$, $\{\bm{w}_0, \bm{\tau}, \bm{p}\}$, $\{\bm{\Phi}_0, \tau_0,\bm{\tau}, \bm{p}\}$, and $\{\bm{\Phi}_k\}$. Then, the variables in each block are alternately  solved by its corresponding sub-problem with the other blocks fixed until the convergence is achieved.

\begin{figure}[t]
  \centering
  \includegraphics[width=0.45\textwidth]{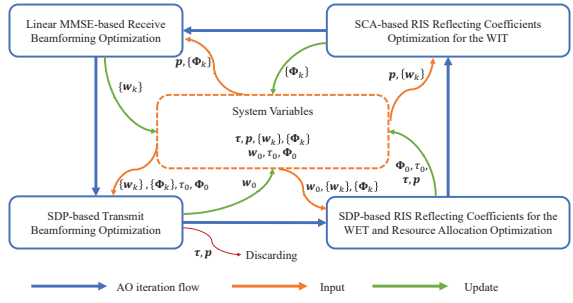}
  \caption{\label{fig:flow}A flow chart of the proposed algorithm.
  }
\end{figure}

\subsection{Linear MMSE-based Receive Beamforming  Optimization}
\label{MMSESection}
With the other variables fixed, we first design the receive beamforming vectors $\{\bm{w}_k\}$. To cope with the interference  caused by $\bm{n}_v$ and $\bm{n}_r$ in \eqref{SNR}, we  apply the linear MMSE criterion here. Based on this criterion, the MMSE-based receive beamforming is given by

  \begin{equation}
    \bm{w}_k^* = \left(\bm{g}_k\bm{g}_k^H + \frac{\sigma_v^2}{p_k}\bm{G}_r\bm{\Phi}_k\bm{\Phi}_k^H \bm{G}_r^H + \frac{\sigma_r^2}{p_k} I_L \right)^{-1} \bm{g}_k, ~\forall k. \label{MMSE}
  \end{equation}

\subsection{SDP-based Transmit Beamforming Optimization}
\label{EBSection}

We then proceed to optimize $\{\bm{w}_0, \bm{\tau}, \bm{p}\}$ with the other variables $\{\bm{w}_k,\bm{\Phi}_0, \tau_0, \bm{\Phi}_k\}$ fixed. Letting $e_k = {p_k}{\tau_k}$ and $\bm{e}= [e_1, \ldots, e_K]$ and  applying the obtained results in \eqref{MMSE}, \textbf{P1} can be simplified as follows
\begin{align}
  (\mbox{\textbf{P2}}) \max_{\substack{ {\bm{w}_0},{\bm{\tau}}, {\bm{e}}}} &\sum_{k=1}^K \tau_k \log \left(1+{\epsilon_k} \frac{e_k}{\tau_k}\right) \label{pro_obj_2}\\
  s.t. \,\,\,\,\,\,&{e_k} \le E_k,~\forall k, \label{E_k2}\\
  & {\tau_k}\ge 0, \, {e_k} \ge 0,~ \forall k, \label{positive_2}\\
  & \mbox{\eqref{time}, \eqref{W0_WET} and \eqref{Pr_WIT},}\nonumber
\end{align}
where $\epsilon_k = \frac{\left|{(\bm{w}_k^*)^H}\left({\bm{G}_r}{\bm{\Phi}_k}{\bm{g}_{u,k}} +{\bm{g}_{d,k}} \right)\right|^2}{\left\|(\bm{w}_k^*)^H{\bm{G}_r}\bm{\Phi}_k\right\|^2 {\sigma_v^2} + {\left\|{\bm{w}_k^*}\right\|^2\sigma_{r}^2}}$, $ \forall k$. It can be found that \textbf{P2} is highly  non-convex.  To solve it efficiently, the SDR technique \cite{luo2010Semidefinite} is employed. Define $\bm{H}_{k} =  \bm{h}_k \bm{h}_k^H $, and $ \bm{G}_{u,k} = \text{diag}\left(\left|\bm{g}_{u,k,1}\right|^2,\left|\bm{g}_{u,k,2}\right|^2,\hdots,\left|\bm{g}_{u,k,N}\right|^2\right) $, $\forall k$. Let  $\bm{W}_0=\bm{w}_0 \bm{w}_0^H $, which satisfies ${\bm{W}_0} \succeq 0$ and $\text{rank}(\bm{W}_0)=1$. Then, \eqref{E_k2} is rewritten as 
\begin{equation}
  {e_k} \le \beta {\tau_0} [{\rm Tr}\left(\bm{H}_{k} \bm{W}_0 \right) + \left\|\bm{h}_{u,k}^H\bm{\Phi}_0\right\|^2 {\sigma_v^2} ], \forall k.\label{E_k3}
\end{equation}
Denote the  RIS reflecting coefficient vector for the WET as $\bm{\varphi}_k =\left[{\phi_{k,1}},{\phi_{k,2}},\hdots,{\phi_{k,N}} \right]^T$, $\forall k$, \eqref{Pr_WIT} can be recast as
\begin{equation}
  {e_k} {\bm{\varphi}_k^H}{G_{u,k}{\bm{\varphi}_k}} + {\tau_k}{\sigma_v^2}{\bm{\varphi}_k^H}{\bm{\varphi}_k}\le {\tau_k}P_r,~\forall k.\label{Pr_WIT2}
\end{equation}
Then, \textbf{P2} can be equivalently transformed into 
\begin{align}
  (\mbox{\textbf{P2-1}})\max_{\bm{W}_0,\bm{\tau}, \bm{e}} \,\,\,\,\,\,&\sum_{k=1}^K \tau_k \log \left(1+{\epsilon_k} \frac{e_k}{\tau_k}\right) \label{pro_w}\\
  s.t. \,\,\,\,\,\,& {\rm Tr}\left(\bm{W}_0 \right) \le P_0, \label{W_0}\\
  & {\bm{W}_0} \succeq 0,\\
  & \text{rank}(\bm{W}_0)=1, \label{W_0_rankone}\\
  &\mbox{\eqref{time}, \eqref{positive_2}, \eqref{E_k3} and \eqref{Pr_WIT2}}.\nonumber
\end{align}

Since the rank-one constraint in \eqref{W_0_rankone} is non-convex, we employ the SDR technique to relax it. Thus,  \textbf{P2-1} becomes to be a convex semidefinite program (SDP) and can be solved with the interior-point method \cite{boyd2004Convex}.
\begin{proposition}
\label{RankOne}
The optimal transmit beamforming matrix obtained by solving the relaxed version of \textbf{P2-1}, denoted by $\bm{W}_0^*$, is rank-one. 
\end{proposition}
\begin{proof}
Please refer to Appendix A.
\end{proof}

According to Proposition \ref{RankOne}, the tightness of SDR is guaranteed. Hence, we can employ Cholesky decomposition to obtain the optimal energy beamforming vector $\bm{w}_0^*$.

\subsection{SDP-based RIS Reflecting Coefficients for the WET and Resource Allocation Optimization}

In this sub-section, we focus on optimizing the reflecting beamforming at the RIS in the WET phase, the transmit power at each user, and the network time scheduling. Since $\tau_0$ and $\bm{\Phi}_0$ are coupled, we first optimize $\{\bm{\Phi}_0, \bm{\tau}, \bm{p}\}$ with $\tau_0$ given.
Define ${\bm{\Psi}_0 = \bm{\tilde{\varphi} _0}\bm{\tilde{\varphi }_0}^H}$ with  $ {\bm{\Psi}_0 } \succeq 0$ and $\text{rank}(\bm{\Psi}_0 )=1$, where ${\bm{\tilde{\varphi }_0} = [{\bm{\varphi _0}^H} , 1]^H }$ and $\bm{\varphi}_0 =\left[{\phi_{0,1}},\hdots,{\phi_{0,N}} \right]^T$. Let $ \bm{H}_{u,k} = \text{diag}\left\{\bm{h}_{u,k,1}, \hdots, \bm{h}_{u,k,N} \right\}$ and $ \bm{Q}_{u,k} = \text{diag}\left\{\left|\bm{h}_{u,k,1}\right|^2, \hdots,\left|\bm{h}_{u,k,N}\right|^2 , 1\right\} $. 
Then,   \eqref{Pr_WET} and  \eqref{E_k2} are respectively reformulated as
\begin{equation}
  {P_0} {\rm Tr}({\tilde{\bm{H}}_r} \bm{\Psi}_0  ) + {\sigma_v^2} {\rm Tr}(\bm{\Psi}_0 )\le P_r,\label{Pr_WET3}
\end{equation}
\begin{equation}
  {e_k} \le \beta {\tau_0}{\rm Tr}[( \bm{V} + {\sigma_v^2} \bm{Q}_{u,k}) \bm{\Psi}_0] -\beta {\tau_0}{\sigma_v^2},~\forall k,\label{E_k4}
\end{equation}
where
$$
\bm{V} = \left[\begin{matrix}
 \bm{H}_{u,k}^H \bm{H}_r \bm{W}_0 \bm{H}_r^H \bm{H}_{u,k}&\bm{H}_{u,k}^H \bm{H}_r \bm{W}_0 \bm{h}_{d,k}\\
 \bm{h}_{d,k}^H \bm{W}_0 \bm{H}_r^H \bm{H}_{u,k}&\bm{h}_{d,k}^H \bm{W}_0 \bm{h}_{d,k}
\end{matrix}\right],
$$
and
$$
   \tilde{\bm{H}}_r ={ \left[\begin{matrix}
    \bm{H}_r \bm{H}_r^H&\bm{0}\\
    \bm{0}&\bm{0}
   \end{matrix}\right]}.
$$

With the obtained solutions in Sections \ref{MMSESection} and \ref{EBSection}, \textbf{P1} can be  equivalently written as 

\begin{align}
  (\mbox{\textbf{P2-2}})\max_{\substack{{\bm{\Psi}_0},\bm{\tau}, \bm{e}}} &\sum_{k=1}^K \tau_k \log \left(1+{\epsilon_k} \frac{e_k}{\tau_k}\right) \label{pro_exp}\\
  s.t. \,\,\,\,\,\,& {\bm{\Psi}_0} \succeq 0,\, \text{rank}(\bm{\Psi}_0)=1, \label{Psi_rankone}\\
  & {[\bm{\Psi}_0]_{n,n}} \le {a_{max}^2},~ \forall n,\\
  & {[\bm{\Psi}_0]_{N+1,N+1}} = 1,  \,\,\,\,\, \\
  & \mbox{\eqref{time}, \eqref{positive_2}, \eqref{Pr_WIT2}, \eqref{Pr_WET3} and \eqref{E_k4}.}\nonumber
\end{align}

Similarly, after the relaxation of the rank-one constraint in \eqref{Psi_rankone}, $\mbox{\textbf{P2-2}}$ is also an SDP and can be solved by the interior-point method.  Recall that the tightness of optimizing $\bm{W}_0$ by SDR can be  guaranteed, we can also prove that the obtained solution $\bm{\Psi}_0$ is rank-one. Then, $\bm{\tilde{\varphi }}_0$ can be recovered by implementing  Cholesky decomposition of $\bm{\Psi}_0$, and the optimal reflection coefficient vector for the WET $\bm{\varphi _0}^*$ can be obtanied by linear operation from $\bm{\tilde{\varphi }}_0$. Subsequently, the optimal RIS reflecting coefficient matrix $\bm{\Phi}_0^*$ can be obtained by $\bm{\Phi}_0^*\,=\, \text{diag}((\bm{{\varphi }}_0^H)^*)$.

Finally, we continue to update the optimal energy transmission time $\tau_0 \in [0,1]$ by the one-dimensional search method. Thus, the maximum sum-rate of this sub-problem is achieved with the optimal solution $\{\bm{\Phi}_0^*, \tau_0^*,\bm{\tau}^*, \bm{p}^*\}$. The procedure is summarized in Algorithm \ref{al:sdp2}.
\begin{figure}[h]
    \begin{algorithm}[H]
    \renewcommand{\algorithmicrequire}{\textbf{Input:}}
    \renewcommand{\algorithmicensure}{\textbf{Output:}}
		\caption{\label{al:sdp2}SDP-based  RIS reflecting coefficients for the WET and resource allocation optimization}
		\begin{algorithmic}[1]
		    \REQUIRE {$\bm{w}_0,\{\bm{w}_k\},\{\bm{\Phi}_k\},~\forall k$}
        \ENSURE {$\bm{\Phi}_0^*,\tau_0^*,\bm{\tau}^*,\bm{p}^*$}
        \STATE{Initialization: The maximum objective function value $R_{\max}= 0$ and the step size $\delta$.}
		    \FOR{$\tau_0 = 0:\delta:1 $}
		    \STATE {Given $\bm{w}_0,\{\bm{w}_k\},\{\bm{\Phi}_k\}$, we obtain ${\tau_{0}^\prime}$, ${\bm{\Psi}_0^\prime}$, ${\tau_k^\prime}$, ${e_k^\prime}$, $~\forall k$ by solving  $\mbox{\textbf{P2-2}}$.}
		    \STATE {Calculate $R\,=\, \sum_{k=1}^K \tau_k^\prime \log \left(1+{\epsilon_k} \frac{e_k^\prime}{\tau_k^\prime}\right)$.}
		    \IF{$R > R_{\max}$}
		    \STATE{Update $R_{max}\leftarrow R$.}
		    \STATE{Update $\tau_0\leftarrow \tau_0^\prime,\bm{\Psi}_0\leftarrow\bm{\Psi}_0^\prime,\tau_k\leftarrow\tau_k^\prime,e_k\leftarrow e_k^\prime$.}
		    \ENDIF
		    \ENDFOR
		    \STATE{Obtain $\bm{\tilde{\varphi }_0}$ from $\bm{\Psi}_0$ by Cholesky decomposition.}
		    \STATE{Obtain $\bm{{\varphi }}_0$ from $\bm{\tilde{\varphi }}_0$.}
		    \STATE{Set $\bm{\Phi}_0^*\,=\, \text{diag}((\bm{{\varphi }}_0^H)^*)$.}
		    \STATE{Calculate $p_k^* \,=\,e_k^*/\tau_k^*$.}
            \RETURN {$\bm{\Phi}_0^*,\tau_0^*,\bm{\tau}^*,\bm{p}^*$.}
		\end{algorithmic}
	\end{algorithm}
\end{figure}

\subsection{SCA-based RIS Reflecting Coefficients Optimization for the WIT}

In this sub-section, we investigate the optimization of RIS reflecting coefficient matrix $\bm{\Phi}_k$ in the WIT phase, which is given by

\begin{align}
  (\mbox{\textbf{P3}}) &\max_{\substack{ \bm{\Phi}_k}} \sum_{k=1}^K R_k \label{pro_phi0}\\
  s.t. \,\,\,\,\,\,& \mbox{\eqref{RE_WIT} and \eqref{Pr_WIT}.}\nonumber
\end{align}

 Note that \textbf{P3} is still a non-convex optimization problem as the active RIS introduces additional noise term in the denominator of the objective function, which results in a quadratic fractional programming problem. In fact, the RIS reflecting coefficients  include amplitude reflection coefficients and phase shifts. To simplify this problem, we derive  the optimal phase shifts in the closed-form and then exploit an SCA algorithm to obtain the near-optimal amplitude reflection coefficients according to \cite{long2021Active}.

It is worth noting that  \textbf{P3} can be decomposed into $K$ independent subproblems, each of which maximizes the SNR of $U_k$ at the RS during $\tau_k$ with respect to the RIS reflection coefficient vector $\bm{\varphi}_k$.
Specifically, with  $\bm{w}_k^*$ obtained in \eqref{MMSE} and introducing some new auxiliary variables, the $k$-th SNR maximization problem can be formulated as
\begin{equation}\label{gamma}
\gamma _k =  \frac{{p_k}\left|{\bm{b}_k^H} {\bm{\varphi}_k} +{g_{d,k}} \right|^2} {{{\bm{\varphi}_k^H}{\bm{Q}_r}{\bm{\varphi}_k}}{\sigma_v^2} + {\sigma_r^2}},
\end{equation}
where  $g_{d,k}=\bm{w}_k^H \bm{g}_{d,k}$,  $\bm{g}_{r,k}^H = \bm{w}_k^H \bm{G}_r$, $\bm{b}_k^H = \bm{g}_{r,k}^H \text{diag}\left(\bm{g}_{u,k} \right)$,  $\bm{Q}_r = \text{diag}\left\{\left|\bm{g}_{r,k,1}\right|^2, \hdots,\left|\bm{g}_{r,k,N}\right|^2\right\}$, $\forall k$. Let  $ \bm{F}_k = {p_k}{\bm{G}_{u,k}} + {\sigma_v^2}{I_N} $, $\forall k$. The subproblem can be expressed as 
\begin{align}
  (\mbox{\textbf{P4}}) \max_{{\bm{\varphi}_k}} \,\,\,&\gamma _k \label{pro_phik}\\
  s.t. \,\,\,\,\,\,& {\bm{\varphi}_k^H}{\bm{F}_k}{\bm{\varphi}_k} \le P_r,  ~\forall k, \label{Pr_WIT3}\\
  & \mbox{\eqref{RE_WIT}.}\nonumber
\end{align}

To solve \textbf{P4}, we decompose the optimization of RIS reflecting coefficient vector $\bm{\varphi}_k$ into two sub-problems for  the  amplitude reflection coefficient design  and the optimal phase shift design,  respectively. Let $\bm{\varphi}_k =\bm{\varTheta}_k \bar{\bm{\varphi}}_k$, where $\bm{\varTheta}_k= \text{diag}\left\{e^{j\theta_{k,1}},\ldots,e^{j\theta_{k,N}}\right\} \in \mathbb{C}^{N\times N}$ and $\bar{\bm{\varphi}}_k = \left[a_{k,1}, \hdots,a_{k,N}\right]^T \in \mathbb{R}^{N\times 1}$. 

\subsubsection{Optimization of  phase shifts for the WIT}
The optimal design of phase shifts is given in the following proposition.

\begin{proposition} 
\label{Prooptimal_ris_phase}
The optimal RIS phase shift of the $n$-th RE for the WIT during $\tau_k$ is derived as
\begin{align}
    & \theta_{k,n}^* = arg(g_{d,k}) - arg(\bm{g}_{u,k,n}) +arg(\bm{g}_{r,k,n}), \\
    & ~~~~~~~~~~~~~~~~~~~~~~~~~~~~\forall k, \forall n \nonumber,
\end{align}
where $\bm{g}_{u,k,n}$ and $\bm{g}_{r,k,n}$ denote the $n$-th element of the vector $\bm{g}_{u,k}$ and $\bm{g}_{r,k}$, respectively.

\end{proposition}

\begin{proof}
Please refer to Appendix B.
\end{proof}

\subsubsection{Optimization of amplitude reflection coefficients for the WIT}
For  \textbf{P4}, the optimal design of phase shifts shown in Proposition \ref{Prooptimal_ris_phase}  holds because the value of the amplification power in \eqref{RE_WIT} and \eqref{Pr_WIT3} and the noise power in the denominator of \eqref{gamma} are independent with the phase shift of each RE. In addition, optimizing $\theta_{k,n}$ is equivalent to maximizing the objective function in \eqref{gamma} \cite{long2021Active}. With the optimal phase shifts in Proposition \ref{Prooptimal_ris_phase}, we proceed to optimize the RIS amplitude reflection coefficients for the WIT. In particular, $\mbox{\textbf{P4}}$ can be simplified as

\begin{align}
  (\mbox{\textbf{P4-1}}) \max_{\bar{\varphi}_k} \,\,\,\,\,\,&\bar{\gamma _k} = \frac{{p_k}\left|{\bar{\bm{b}}_k^H} \bar{\bm{\varphi}}_k +\left|g_{d,k}\right| \right|^2} {{{\bar{\bm{\varphi}}_k^H}{\bm{Q}_r}{\bar{\bm{\varphi}}_k}}{\sigma_v^2} + {\sigma_r^2}}  \label{pro_P32}\\
  s.t. \,\,\,\,\,\,& {\bar{\bm{\varphi}}_k^H}{\bm{F}_k}{\bar{\bm{\varphi}}_k} \le P_r,~\forall k, \label{Pr_WIT4}\\
  & a_{k,n} \le {a_{n,\max}},~\forall k,~\forall n,\label{RE_WIT2}
\end{align}
where $\bar{\gamma}_k= \left|\gamma_k\right|$, and $\bar{\bm{b}}_k$ is element-wise modulus of $\bm{b}_k$. To deal with the non-convexity of the objective function \eqref{pro_P32}, we introduce a new auxiliary variable ${n_k}={{\bar{\bm{\varphi}}_k^H}{\bm{Q}_r}{\bar{\bm{\varphi}}_k}}{\sigma_v^2} + {\sigma_r^2}$, which denotes the noise power received at the RIS. Then, \textbf{P4-1} can be  converted into the following equivalent form
\begin{align}
   (\mbox{\textbf{P4-2}}) \max_{\bar{\gamma}_k,n_k,\bar{\bm{\varphi}}_k} & \,\,\bar{\gamma}_k \label{pro_P33}\\
  s.t.  & \sqrt{p_k}\left(\bar{\bm{b}}_k^H \bar{\bm{\varphi}}_k + \left|g_{d,k}\right|\right)\ge \sqrt{n_k \bar{\gamma}_k},~\forall k, \label{SCA1}\\
  & \bar{\bm{\varphi}}_k^H \bm{Q}_r \bar{\bm{\varphi}}_k \sigma_v^2 + \sigma_r^2 \le n_k,~\forall k, \label{noise}\\
  & \mbox{\eqref{Pr_WIT4} and \eqref{RE_WIT2}}.\nonumber
\end{align}

However, the constraint \eqref{SCA1} is still non-convex. To solve \textbf{P4-1}  efficiently, we exploit the SCA algorithm  to approximate the square root by a convex upper-bound in each iteration. Define ${\bar{\gamma}_k(t)}$ and $ {n_k(t)}$ as the iterative optimization variables after the $t$-th step iteration. In terms of $\left\{\bar{\gamma}_k\left(t\right), n_k\left(t\right)\right\} $, the first-order Taylor polynomial is used to approximate ${\sqrt{n_k \bar{\gamma}_k}}$, which is given by
\begin{equation}
\label{SCA}
    \begin{aligned}
      \sqrt{n_k \bar{\gamma}_k} \le &\mathcal{G} (\bar{\gamma}_k, n_k;t) \\
      = &\sqrt{\bar{\gamma}_k(t) n_k(t)}+\frac{1}{2}\left(\frac{n_k(t)}{\bar{\gamma}_k(t)}\right)^{\frac{1}{2}}\left[\bar{\gamma_k}-\bar{\gamma}_k(t)\right]
      \\&+\frac{1}{2}\left(\frac{\bar{\gamma}_k(t)}{n_k(t)}\right)^{\frac{1}{2}}\left[n-n_k(t)\right].
    \end{aligned}
\end{equation}
  
Based on  \eqref{SCA}, \eqref{SCA1} can be rewritten as
\begin{equation}
  \sqrt{p_k}\left(\bar{\bm{b}}_k^H \bar{\bm{\varphi}}_k + \left|g_{d,k}\right|\right)\ge \mathcal{G} (\bar{\gamma}_k, n_k;t),~\forall k.\label{SCA2}
\end{equation}

Then, $\textbf{P4-2}$ can be  reformulated as the following  problem
\begin{align}
  (\mbox{\textbf{P4-3}}) \max_{\bar{\gamma}_k,n_k,\bar{\bm{\varphi}}_k} & \,\,\bar{\gamma}_k, \label{pro_P3SCA}\\
  s.t.  &  \mbox{\eqref{Pr_WIT4}, \eqref{RE_WIT2}, \eqref{noise} and \eqref{SCA2}.}\nonumber
\end{align}

As \textbf{P4-3} is convex and  can be solved by the inter-point method. We then discuss the initialization of ${\bar{\gamma}_k(t)}$ and $ {n_k(t)}$. First, we propose an initial solution $\bar{\bm{\varphi}}_k(0)$ by solving a simple feasible version of problem \textbf{P4-1}, i.e., $\bar{\bm{\varphi}}_k(0)$, satisfying the constraints \eqref{Pr_WIT4} and \eqref{RE_WIT2}. Then, the reasonable initialization of ${\bar{\gamma}_k(0)}$ and $ {n_k(0)}$ is given by
\begin{align}
  \bar{\gamma}_k(0) &= \frac{{p_k}\left|{\bar{\bm{b}}_k^H} \bar{\bm{\varphi}}_k(0) +\left|g_{d,k}\right| \right|^2} {{{\sigma_v^2}\bar{\bm{\varphi}}_k^H(0){\bm{Q}_r}\bar{\bm{\varphi}}_k}(0) + {\sigma_r^2}}, \label{init_gamma}\\
  n_k(0)&={\sigma_v^2}\bar{{\bm{\varphi}}}_k^H(0) \bm{Q}_r \bar{\bm{\varphi}}_k(0) + {\sigma_r^2}.\label{init_n}
\end{align}

With the initialization described in \eqref{init_gamma} and \eqref{init_n}, the optimal amplitude reflection coefficients for the WIT, denoted by $\bar{\bm{\varphi}}_k^*$,  can be obtained by iteratively solving \textbf{P4-3} until the convergence is achieved. 
 As a result, the RIS reflecting coefficients during $\tau_k$ can be calculated by
\begin{equation}
    \bm{\Phi}_k^* = \text{diag}\left(\bm{\varTheta}_k^* \bar{\bm{\varphi}}_k^* \right), ~\forall k,
\end{equation}
where $\bm{\varTheta}_k^* = \text{diag} \{e^{j \theta_{k,1}^*},\ldots, e^{j \theta_{k,N}^*} \}$.

The detailed description of optimizing the RIS reflecting coefficients in  $\mbox{\textbf{P3}}$ is summarized in the Algorithm \ref{al:sca}.

\begin{figure}[h]
    \begin{algorithm}[H]
    \renewcommand{\algorithmicrequire}{\textbf{Input:}}
    \renewcommand{\algorithmicensure}{\textbf{Output:}}
		\caption{\label{al:sca}SCA-based RIS reflecting coefficients for the WIT}
		\begin{algorithmic}[1]
		    \REQUIRE{$\{\bm{w}_k\}, ~\bm{p}, ~ \forall k $}
        \ENSURE {$\{\bm{\Phi}_k^*\}, ~ \forall k$}
	        \STATE {Initialization: $\bar{\gamma}_k(t)$, $n_k(t)$,  $\bar{\bm{\varphi}}_k(t)$, and $t=0$.}
            \STATE {Obtain $\theta_{k,n}^*$ in Proposition \ref{Prooptimal_ris_phase} and have $\bm{\varTheta}_k^* $.}
      \REPEAT
      \STATE{$t=t+1$.}
			\STATE {Update $\bar{\gamma}_k(t)$, $n_k(t)$, $\bar{\bm{\varphi}}_k(t)$ by solving \textbf{P4-3}.}
			\UNTIL{the convergence is achieved.}
			\STATE {Obtain $\bm{\Phi}_k^* = \text{diag}\left(\bm{\varTheta}_k^* \bar{\bm{\varphi}}_k^* \right)$, where $\bar{\bm{\varphi}}_k^* = \bar{\bm{\varphi}}_k(t)$.}
			\RETURN {$\{\bm{\Phi}_k^*\}$.}
		\end{algorithmic}
	\end{algorithm}
\end{figure}

\subsection{Algorithm Summarization and Analysis}

Based on the above analysis, the algorithm for solving \textbf{P1} is summarized in Algorithm \ref{al:ao}. Based on the optimality analysis,the objective function of \textbf{P1} is a non-decreasing function. Due to the power budget constraint \eqref{W0_WET}, \eqref{Pr_WET} and \eqref{Pr_WIT}, the optimal objective value of problem \textbf{P1} is bounded. Hence, the convergence of Algorithm \ref{al:sca} can be thus guaranteed, which will be also confirmed by numerical simulations in Section \ref{Simulation}.

\subsubsection{Complexity Analysis}
The computational complexity of our proposed AO algorithm is analyzed as follows, which contains the linear MMSE-based receive beamforming calculation, the SDR algorithm and the SCA algorithm in each iteration. For the linear MMSE-based receive beamforming optimization, we derive a closed-form solution to the \eqref{MMSE}, and the approximate worst-case computational complexity is given by $\mathcal{O}\left(KL\max(N,L)^2 \right)$. According to \cite{luo2010Semidefinite}, for the subproblem of SDP-based transmit beamforming optimization, the worst-case computational complexity is $\mathcal{O}\left( \max(K,M)^{4.5}\log(1/ \epsilon ) \right)$, where $\epsilon $ is the computational accuracy of the interior-point method in CVX. Similarly, for the SDP-based RIS reflecting coefficients for the WET and resource allocation optimization, the worst-case computational complexity is $\mathcal{O}\left( \mathcal{I}_\tau \max(N,K)^{4.5}\log(1/ \epsilon) \right)$, where $\mathcal{I}_\tau$ is the iteration number for updating $\tau_0$. For the SCA-based RIS reflecting coefficients optimization in the WIT, the computational complexity is less than $\mathcal{O}\left( \mathcal{I}_{S}KN^{3.5} \log(N/ \epsilon) \right)$, where $\mathcal{I}_{S}$ is the iteration number for the SCA algorithm. Thus, the computational complexity of the overall AO algorithm is given by
\begin{equation}
\begin{aligned}
  \mathcal{O}\left(\mathcal{I}_{A} \left(KL\max(N,L)^2+\max(K,M)^{4.5}\log(1/ \epsilon) \right.\right.\\
  \left.\left. + \mathcal{I}_\tau \max(N,K)^{4.5}\log(1/ \epsilon)+\mathcal{I}_{S}KN^{3.5} \log(N/ \epsilon) \right)\right),
\end{aligned}
\end{equation}
where $\mathcal{I}_{A}$ denotes the number of iterations required for convergence.

\subsubsection{Optimality Analysis}
As the formulated problem \textbf{P1} is extremely non-convex, it is very difficult to obtain the globally optimal solution. To solve \textbf{P1} efficiently, we propose an efficient AO algorithm to obtain the suboptimal solutions. Firstly, we obtain the optimal receive beamforming $\left\{\bm{w}_k\right\}$ in a closed-form, which are the globally optimal solutions. With the obtained receive beamforming solutions, the formulated problem can be simplified but is still non-convex. Secondly, the SDR technique is adopted to optimize the transmit beamforming, the RIS reflecting coefficient matrices for the WET phase, the transmit power at each user, and the network time scheduling. In particular, we prove the obtained solutions of \textbf{P2-1} and \textbf{P2-2} are rank-one. Since the tightness of applying SDR can be guaranteed, the obtained solutions $\{{\bm{w}_0,\bm{\Phi}_0,\bm{\tau},\bm{p}}\}$ are globally optimal \cite{hua2022PowerEfficient}. Then, we use the one-dimensional search method to exploit the optimal energy transmission time $\tau_0$ by setting an appropriate step size. Thirdly, for the optimization of RIS reflecting coefficients for the WIT phase, we decompose \textbf{P4} into two sub-problems. On the one hand, the optimal phase shifts have been derived in a closed-form which has been proved. On the other hand, \textbf{P4-1} is solved by Algorithm \ref{al:sca}, which obtains the reflection amplitudes of $\{\bm{\Phi}_k\}$ are near-optima of the original problem \cite{razaviyayn2014Successive}.  Hence, Algorithm \ref{al:ao} can be used to obtain the near-optimal solutions to \textbf{P1} with a high accuracy.

\begin{figure}[h]
    \begin{algorithm}[H]
    \caption{\label{al:ao}AO algorithm for \textbf{P1}}
    \begin{algorithmic}[1]
        \STATE {Initialization: $\bm{w}_0,\{\bm{w}_k\},\bm{\Phi}_0,\{\bm{\Phi}_k\},\tau_0,\bm{\tau},\bm{p},~\forall k$.}
        \REPEAT
            \STATE {Given $\bm{\Phi}_k$ and $\bm{p}$, update $\{\bm{w}_k\}$  by \eqref{MMSE}.}
              \STATE {Given $\bm{\Phi}_0$, $\tau_0$, $\{\bm{\Phi}_k\}$, $\{\bm{w}_k\}$, update $\bm{w}_0$ with  by solving \textbf{P2-1}.} 
              \STATE {Given $\bm{w}_0$, $\{\bm{w}_k\}$, $\{\bm{\Phi}_k\}$, update $\bm{\Phi}_0$, $\tau_0$, $\bm{\tau}$ and $\bm{p}$ by Algorithm \ref{al:sdp2}.}
              \STATE {Given $\{\bm{w}_k\}$ and $\bm{p}$, update \{$\bm{\Phi}_k\}$  by Algorithm \ref{al:sca}.}
            \UNTIL {$\sum_{k=1}^K R_k$ converged.}
            \RETURN {$\bm{w}_0^*,\{\bm{w}_k^*\},\bm{\Phi}_0^*,\{\bm{\Phi}_k^*\},\tau_0^*,\bm{\tau}^*,\bm{p}^*$.}
    \end{algorithmic}
  \end{algorithm}
\end{figure}

\section{Numerical Results}
\label{Simulation}

\begin{figure}[t]
  \centering
  \includegraphics[width=0.45\textwidth]{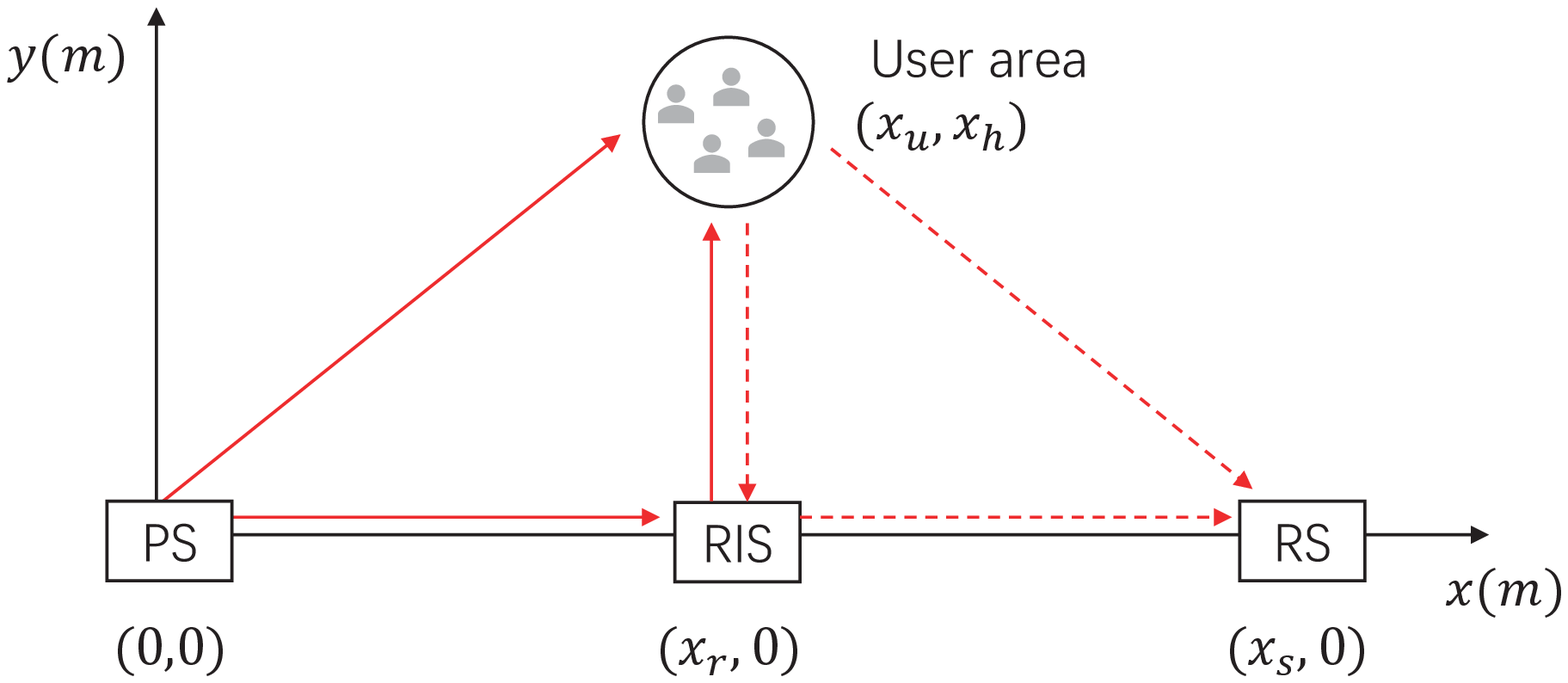}
  \caption{\label{fig:placement}Placement model of simulation setup.
  }
\end{figure}

In this section, numerical results are presented to evaluate the performance of the proposed scheme. As shown in Figure \ref{fig:placement}, we consider that the simulated network deployment is a 2-D coordinate system, where the coordinates of the PS, the RIS, and the RS are given as (0,0), ($x_r$, 0), and ($x_s$,0), respectively, the users are randomly deployed within a circular area centered at ($x_u$, $x_h$) with radius 1m. We follow the channel model  considered in \cite{zeng2022Throughput}. In particular, the large-scale path-loss is modeled as $L=A(d/d_0)^{-\alpha}$, where $A$ is the path-loss at the reference distance $d_0 = 1$m and set as $A = -30 $dB, $d$ denotes the distance between two nodes, and $\alpha$ is the path-loss exponent. For the RIS related links, the path-loss exponent is set as 2.2 since the location of RIS can be carefully designed to avoid the severe signal blockage. While  the path-loss exponents for the RIS unrelated links are set as 3.5 due to the users' random deployment. We assume the direct link channels follow Rayleigh fading but the RIS related channels follow Rician fading. Specifically, the small-scale channel from the PS to the RIS can be expressed as 
\begin{equation}
\bm{H_r}=\left(\sqrt{ \frac{\beta_{r}}{\beta_{r}+1}}\bar{\bm{H}}_r^\text{LoS} + \sqrt{ \frac{1}{\beta_{r}+1}} \bar{\bm{H}}_r^\text{NLoS} \right)
\end{equation}
where  $ \beta_{r}$ is the Rician factor for the PS-RIS link, $\bar{\bm{H}}_r^\text{LoS}$ denotes the deterministic line of sight (LoS) component, and $\bar{\bm{H}}_r^\text{NLoS}$ denotes the non-LoS compotent with circularly symmetric complex Gaussian random variables with zero mean and unit variance. The other channels can be similarly defined.
Unless otherwise stated, other parameters are given as follows:  $\beta_r = 10$ \cite{guo2020Weighted}, $\rho = 0.8$, $\sigma_v^2 = \sigma_r^2 = -90$dBm, $P_0 = 20$dBm \cite{long2021Active}, $ P_r = 20$dBm, $a_{max} = 25$dB \cite{amato2018Tunneling}, $ N = 10$, $K = 4$, $M=4$, $L=4$, $x_r = 10$m, $x_u = 10$m, $x_s = 20$m, and $x_h = 2$m. 

For comparisons, we also evaluate the performance of the following  benchmark schemes:

\begin{itemize}
\item[(1)] Active RIS-aided single-antenna WPCN scheme (Active-SA). \item[(2)] Passive RIS-aided multi-antenna WPCN scheme (Passive-MA). \item[(3)] Active RIS-aided multi-antenna WPCN with uniform energy beamforming scheme (Active-MA-UEBF).
\end{itemize}

Notice that for  the multi-antenna schemes, the number of antennas is set as  4 in the PS and RS. In addition, we set the number of REs in the Passive-MA scheme as $N=100$ to show the superiority of the proposed scheme. 

Before performance comparisons, we first show the convergence performance of the proposed AO algorithm in Figure \ref{fig:converge}. One can  observe that as the number of iterations increases, the sum-rate first increases  but finally converges to a constant after nearly 8 iterations. This demonstrates
that the convergence of the proposed scheme can be achieved quickly. The other observation is that the effect of the parameter setting on convergence is limited. 

\begin{figure}[ht]
   \centering
   \includegraphics[width=0.45\textwidth]{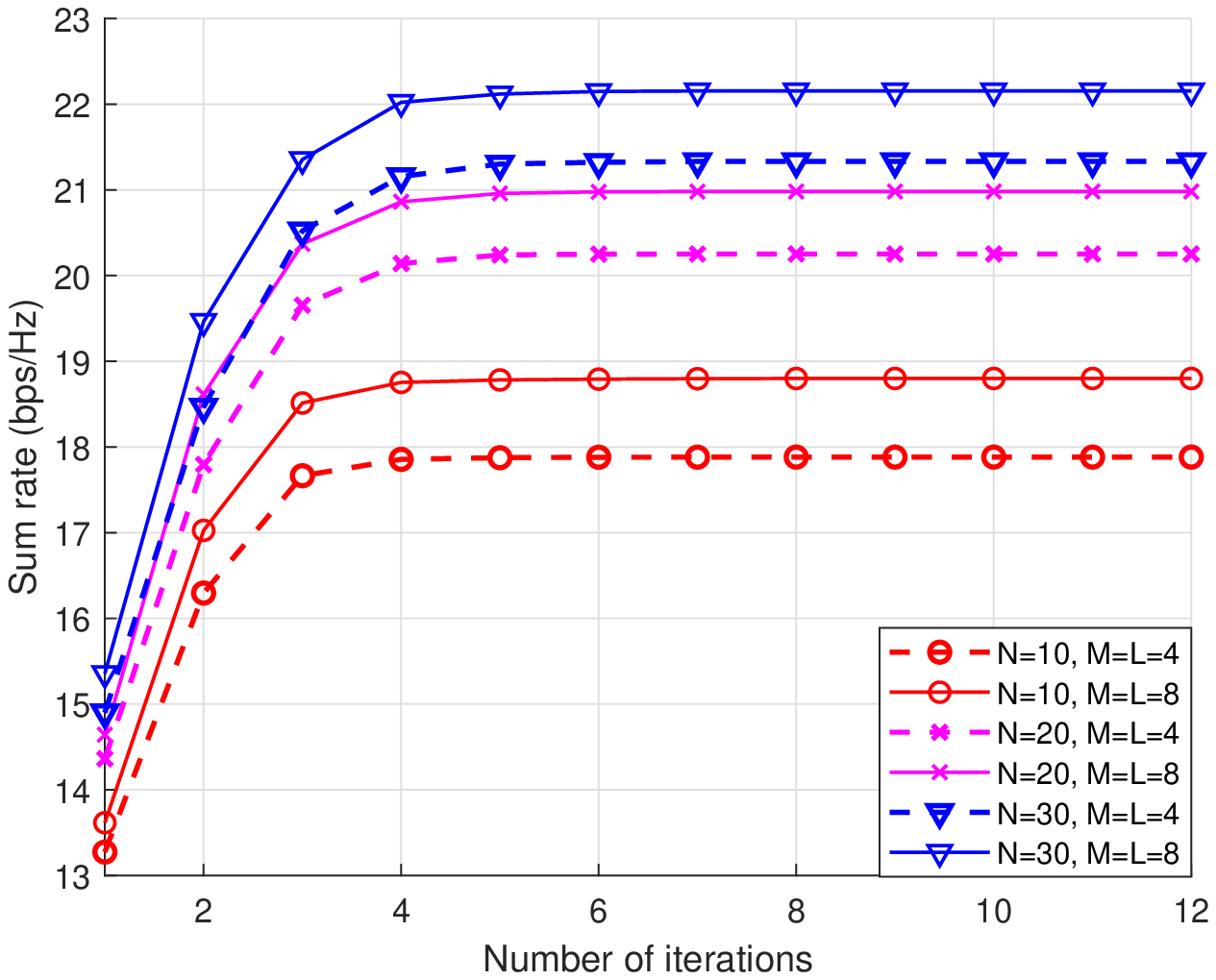}
   \caption{\label{fig:converge} Convergence behavior of the proposed scheme under different parameter settings.   }
\end{figure}

Figure \ref{fig:p0} shows the impact of the transmit power at the PS (i.e., $P_0$) on  the sum-rate when the RIS's maximum reflecting power $P_r = 10,\,20$ dBm and the RIS's maximum amplitude reflection coefficient $a_{\max} = 10,\,25$ dB, respectively. In general, the proposed scheme outperforms the Active-SA scheme with the same parameters, which confirms that the assistance of  multiple antennas can achieve a significant performance gain by constructing the transmit beamforming at the PS and the receive beamforming at the RS. For a given $a_{max} = 25$ dB, our proposed scheme with 10 REs can achieve 415.48\% performance gain compared to the passive RIS scheme with 100 REs when $P_0 = 20$ dBm. Indeed, the active RIS can considerably make use of its amplification characteristic to amplify the energy signals at low transmit power and thereby realize a superior capability at the cost of additional power consumption. For a given $a_{max} = 10$ dB, it can be seen that the sum-rates achieved by the proposed scheme with $P_r = 20$ dBm and $P_r = 10$ dBm are almost the same, which implies that the amplification power constraints defined in \eqref{Pr_WET} and \eqref{Pr_WIT} are inactive since $a_{\max}$ is limited for the small transmit power. Note that, the performance gap is significant between the scheme with $a_{\max} = 25$ dB and $a_{\max} = 10$ dB because $a_{\max}$ directly limits the amplitude reflection coefficient of the active RIS. In addition, the sum-rate of the passive-MA scheme is generally lower than the active RIS schemes with the same REs and the passive RIS needs to be equipped with more REs (e.g., 100 REs) to achieve the similar performance.

\begin{figure}[t]
  \centering
  \includegraphics[width=0.45\textwidth]{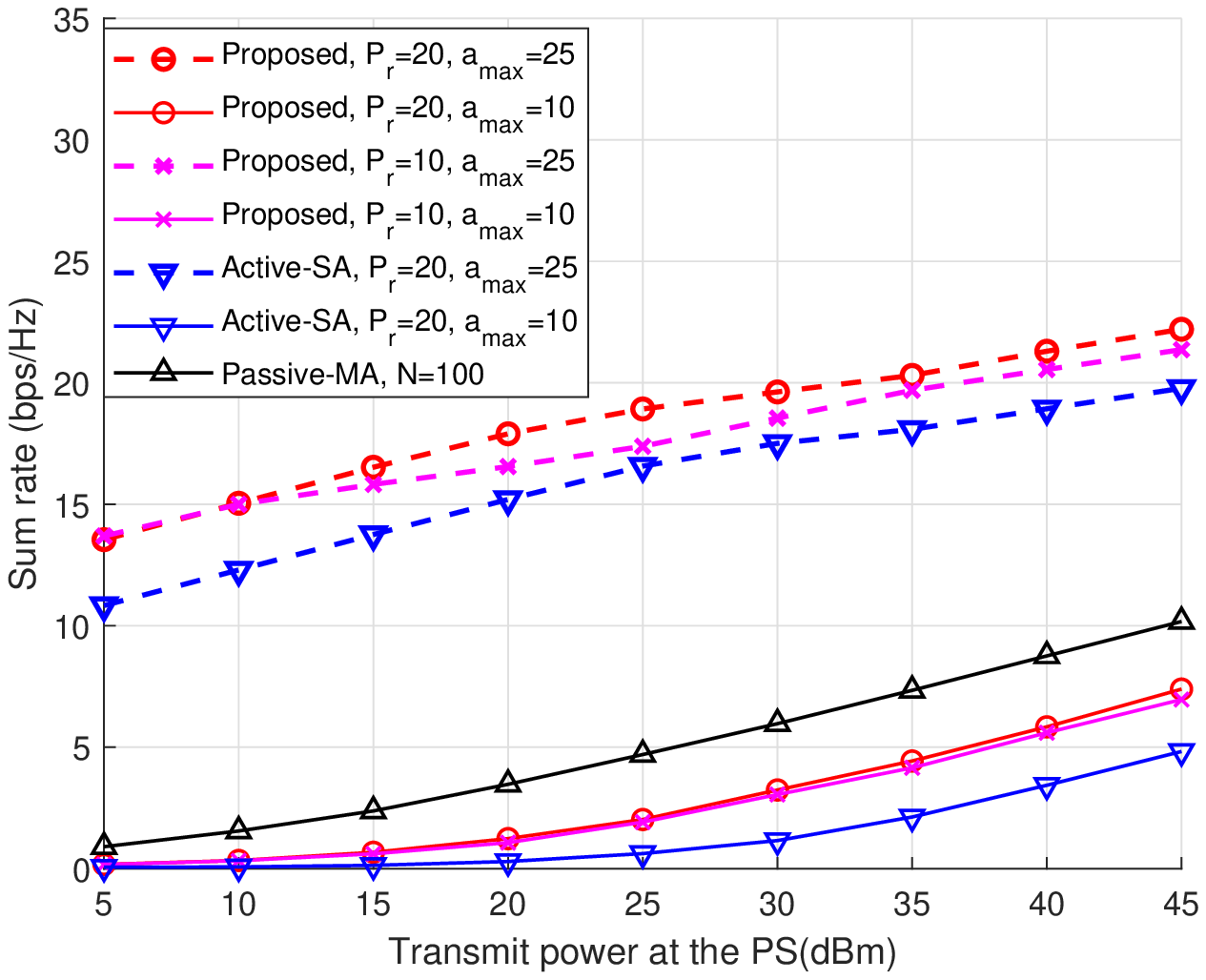}
  \caption{\label{fig:p0} Sum-rate versus the transmit power at the PS. }
\end{figure}

\begin{figure}[t]
  \centering
  \includegraphics[width=0.45\textwidth]{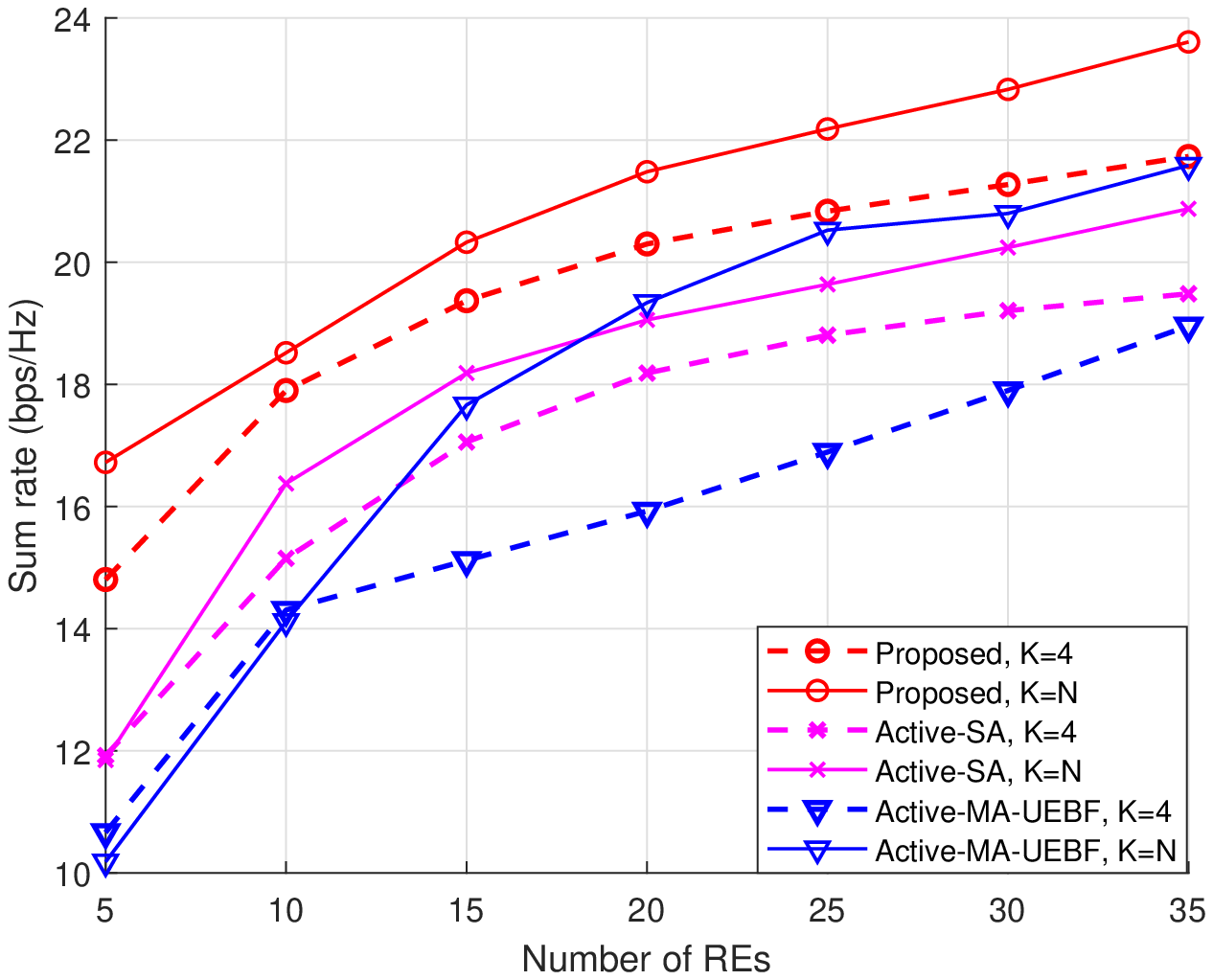}
  \caption{\label{fig:fig_element}Sum-rate versus number of reflecting elements at the RIS.}
\end{figure}

In Figure \ref{fig:fig_element}, we evaluate the sum-rate versus the number of reflecting elements at the RIS. It can be seen that the proposed schemes can achieve a higher performance gain compared with the other benchmark schemes. With an increasing number of reflecting elements, the sum-rate increases due to the fact that more transmission links can be provided for both the WET and the WIT. In addition, to investigate the best system performance, the maximum number of users is set to be equal to the number of REs. Since the active RIS can amplify the incident signals, a limited number of REs is sufficient to reach the desired SNR. Therefore, the size of active RIS can be reduced, making it applicable to the scenario where the space for the RIS deployment is limited.

\begin{figure}[t]
  \centering
  \includegraphics[width=0.45\textwidth]{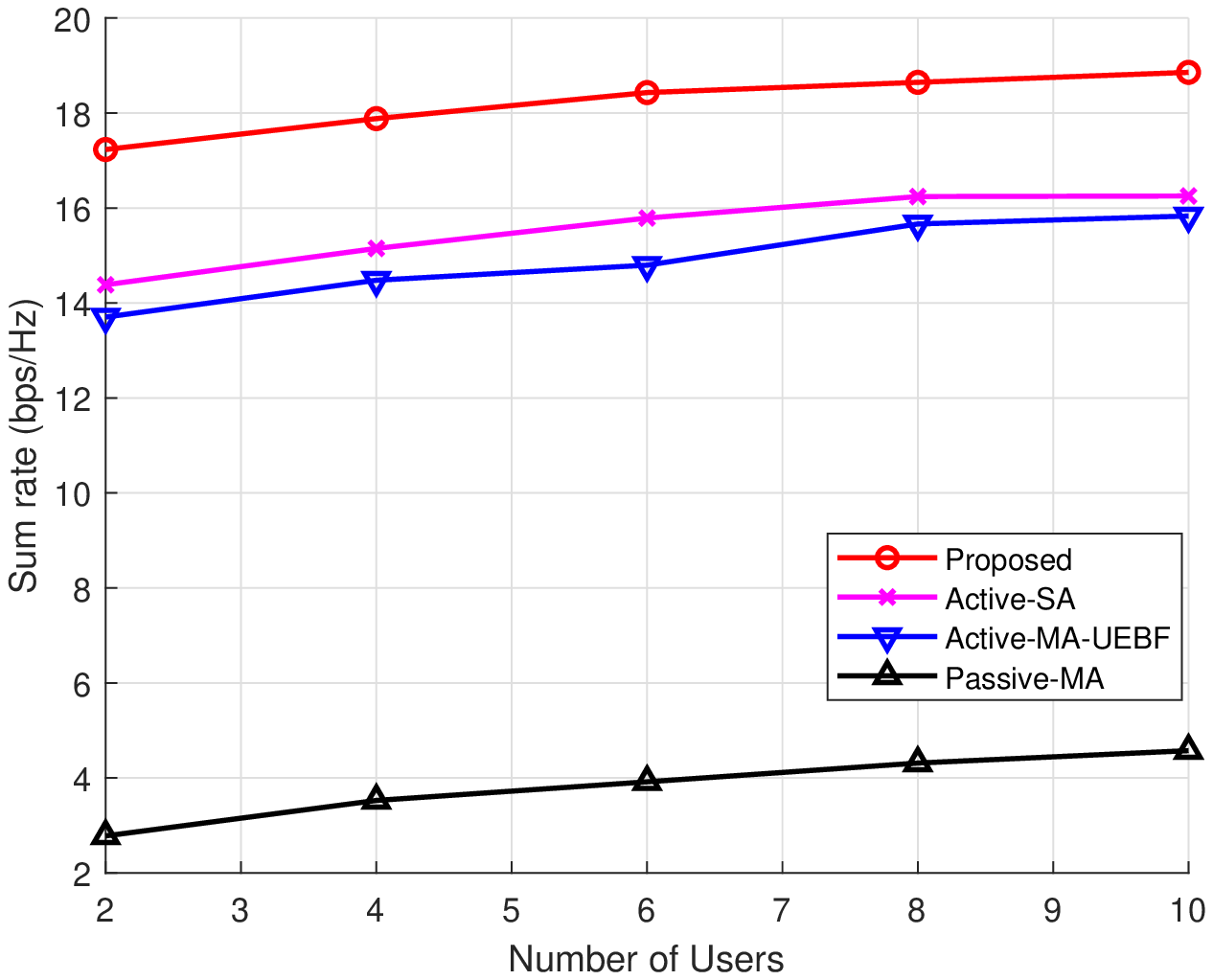}
  \caption{\label{fig:usern100} Sum-rate versus the number of users. }
\end{figure}

In Figure \ref{fig:usern100}, we study the effect of number of user on the sum-rate.  As the number of users increases, the total amount of harvested energy by users improves, which results in a higher sum-rate. Nonetheless, when the number of users reaches a threshold, e.g. $K=8$, the sum-rate achieved by our proposed scheme becomes to be saturated. This is due to the fact that the increment of number of users reduces the energy transfer duration and the time allocated to each user for information transmission, which makes the sum-rate converge to a constant. Again, our proposed scheme notably outperforms the other benchmark schemes. For example, when the number of users is $K=4$, our proposed scheme can achieve 17.78\% and 415.48\%  performance gain compared with the Active-SA scheme and the Passive-MA scheme with 100 REs, respectively.

\begin{figure}[t]
  \centering
  \includegraphics[width=0.45\textwidth]{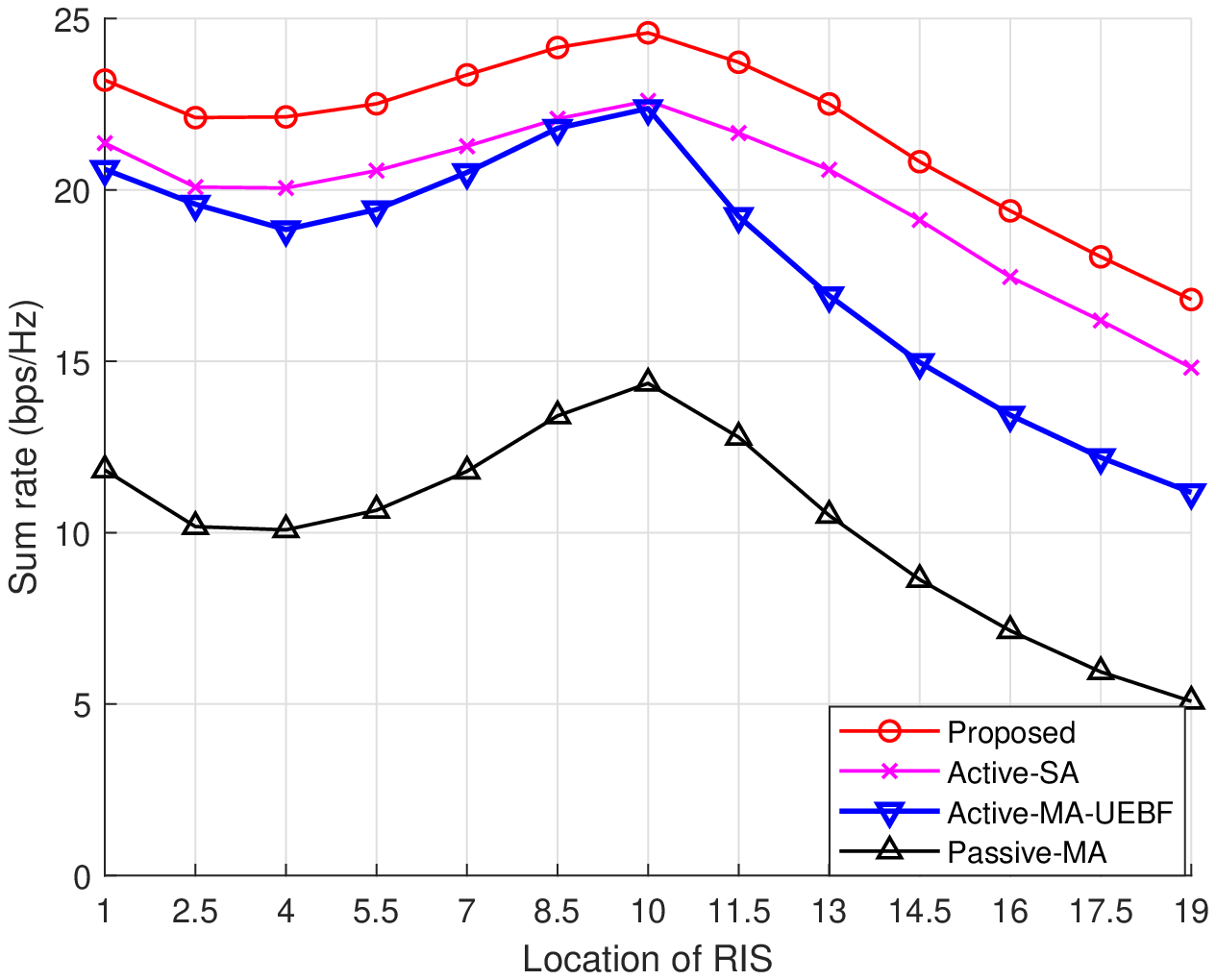}
  \caption{\label{fig:location}Sum-rate versus x-coordinate of the RIS.  }
\end{figure}

In Figure \ref{fig:location}, we plot the sum-rate versus the horizontal ordinate of the RIS. As  $x_r$ varies, the sum-rates of all schemes first increase but then decrease. Compared to the scenario that the RIS is close to the RS, by deploying the RIS near the PS, e.g., $x_r = 1$, the sum-rate can be improved. It is because the users can harvest more energy assisted by the active RIS. Moreover, we can observe that the sum-rate is maximized at $x_r =10$, where the reflecting link between the active RIS and each user is strongest so the users can benefit from a larger amplification and reflection gain. However, when the RIS is neither close to the PS nor the users, both the PS-RIS link and the RIS-users links become weak, which results in the reduce of harvested energy. Furthermore, since the Active-MA-UEBF scheme adopts the uniform energy beamforming, the energy signals cannot adaptively align with the direction of the desired channels, which results in a low WET efficiency. In addition, the schemes with the active RIS can achieve a much better performance than the passive RIS scheme, which demonstrates that the active RIS with the amplification functionality can significantly mitigates the double-fading effect. The above observation demonstrates that the location of the active RIS should be carefully designed.

\section{CONCLUSIONS}
In this paper, we have proposed an active RIS assisted relaying scheme to enhance the performance of multiuser multi-antenna WPCN, which is involved in both the WET from the PS to users and the WIT from users to the RS. To further enhance system performance, both transmit beamforming  at the PS and receive beamforming at the RS have been designed.  We have formulated a system sum-rate maximization problem by jointly optimizing the RIS reflection coefficients for both the WET and the WIT, transmit and receive beamforming vectors, transmit power at each user, and network time scheduling. As the formulated problem is non-convex, we have proposed an AO algorithm with linear MMSE, SDR and SCA techniques to solve it efficiently. Finally, numerical results have been provided to confirm the performance superiority of the proposed scheme.


\section*{Appendix}
\subsection*{A. Proof of Proposition 1}
\label{appendix:rank-one}
The Lagrangian function of \textbf{P2-1} can be expressed as
  \begin{equation}\label{pro_lag}
    \begin{aligned}
  \mathcal{L} =& \sum_{k=1}^K \lambda_k \beta \text{Tr} (\bm{H}_{d,k} \bm{W}_0) -\xi \text{Tr} (\bm{W}_0)\\  
  &+ \text{Tr} (\bm{\Omega}\bm{W}_) + \delta,~\forall k,
    \end{aligned}
  \end{equation}
  where $\lambda_k \ge 0$, $\xi \ge 0$, and $\bm{\Omega}\in \mathbb{H}^{M} $  are the Lagrange multipliers associated with constraints \eqref{E_k3}, \eqref{Pr_WIT2}, and \eqref{W_0}, respectively,  $\delta$ denotes the term unrelated with $\bm{W}_0$. The Karush-Kuhn-Tucker (KKT) conditions  of \textbf{P2-1} are given as follows
  \begin{align}
    &\frac{\partial{\mathcal{L}}} {\partial{\bm{W}_0}} = \sum_{k=1}^K \lambda_k^* \beta \bm{H}_{d,k} - \xi^* \bm{I}_M + \bm{\Omega}^* = 0,  \label{pro_lag_a}\\
    & \bm{\Omega } ^* \bm{W}_0^*=0 \label{pro_lag_b},
  \end{align}
  where $\lambda_k^*$, $\xi^*$ and $\bm{\Omega}^*$ are the optimal Lagrangian multipliers for the dual problem of $\mbox{\textbf{P2-1}}$.  It can be proved that $\lambda_k^* > 0$ and $\xi^* >0$ since  the constraints \eqref{E_k3} and \eqref{Pr_WIT2}are equalities in the optimal condition. Based on \eqref{pro_lag_a} and \eqref{pro_lag_b}, it is straightforward to obtain the following equality
  \begin{align}
  \label{KKTEquality}
  (\xi^* \bm{I}_M - \sum_{k=1}^K \lambda_k^* \beta \bm{H}_{d,k}  ) \bm{W}_0^*=0
  \end{align}

  According to \cite{xu2021Reconfigurable}, $\text{rank} (\xi^* \bm{I}_M - \sum_{k=1}^K \lambda_k^* \beta \bm{H}_{d,k}) \ge M-1$ due to the fact that  $\bm{H}_{d,k}$ for $\forall k$ are independently distributed. Thus, from \eqref{KKTEquality}, we can obtain that $\text{rank}(\bm{W}_0) \le 1$. It is obvious that $\bm{W}_0 = \bm{0}$ is not the optimal solution to \textbf{P2-1}. Hence, we  derive that $\text{rank}(\bm {W}_0) = 1$, which thus proves Proposition \ref{RankOne}.

  \subsection*{B. Proof of Proposition 2}
  \label{appendix:phase}
  Since $\bm{Q}_r$ and $\bm{F}_k$ are diagonal matrices, we observe  that the noise power in the denominator of \eqref{pro_phik} and the amplification power in \eqref{RE_WIT} and \eqref{Pr_WIT3} are  independent of the phase shift of each RE. Therefore, maximizing $\gamma_k$ with respect to $\bm{\varTheta}_k$ is equivalent to the following optimization problem
    \begin{align}
      (\mbox{\textbf{P4-4}})\max_{{\bm{\varTheta}_k}} \,\,\,\,\,\,&\left|{\bm{b}_k^H} \bm{\varTheta}_k{\bar{\bm{\varphi}}_k} +{g_{d,k}} \right|^2 \label{pro_equiv_phase}\\
      s.t. \,\,\,\,\,\,& \left|\varTheta_{k,n}\right|=1,~\forall k , ~\forall n. 
    \end{align}
  
We rewrite the objective function as
    \begin{equation}
      \left|\bm{b}_k^H \bm{\varTheta}_k \bar{\bm{\varphi}}_k\right|^2 + \left|{g_{d,k}}\right|^2 + 2\left|\bm{b}_k^H \bm{\varTheta}_k \bar{\bm{\varphi}}_k\right|\left|{g_{d,k}}\right|\cos \alpha,
    \end{equation}
  where $\alpha = \arctan \frac{\mathrm{Im}(\bm{b}_k^H \bm{\varTheta}_k \bar{\bm{\varphi}}_k)}{\mathrm{Re}(\bm{b}_k^H \bm{\varTheta}_k \bar{\bm{\varphi}}_k)} - \arctan \frac{\mathrm{Im}({g_{d,k}})}{\mathrm{Re}({g_{d,k}})}$. Obviously, the maximum of $\left|{\bm{b}_k^H} \bm{\varTheta}_k{\bar{\bm{\varphi}}_k} +{g_{d,k}} \right|^2$ is achieved when $arg(\bm{b}_k^H \bm{\varTheta}_k \bar{\bm{\varphi}}_k)=arg(g_{d,k}) \triangleq \omega$. Let $v_k = [v_{k,1}, v_{k,2}, ...,v_{k,N}]^T \in \mathbb{R}^{N\times 1} $ and $\bm{\xi}_k = \text{diag}(\bm{b}_k^H)\bar{\bm{\varphi}}_k$. As ${\bm{b}_k^H} \bm{\varTheta}_k{\bar{\bm{\varphi}}_k} = \bm{v}_k^H\bm{\xi}_k$,  \textbf{P4-4} can be rewritten as 
    \begin{align}
      (\mbox{\textbf{P4-5}})\max_{{\bm{v}_k}} \,\,\,\,\,\,&\left|\bm{v}_k^H\bm{\xi}_k\right|^2 \label{pro_equiv_phase_2}\\
      s.t. \,\,\,\,\,\,& \left|\bm{v}_{k,n}\right|=1,~\forall k,~\forall n,\\
      & arg(\bm{v}_k^H\bm{\xi}_k) = \omega,~\forall k. 
    \end{align}
  
  Based on \cite{wu2019Intelligent}, the optimal solution to \textbf{P4-5} can be expressed as $\bm{v}_{k}^* = e^{j(\omega-arg(\bm{\xi}_{k}))} = e^{j(\omega-arg(\text{diag}(\bm{b}_k^H)\bar{\bm{\varphi}}_k))}$. Then, the optimal RIS phase shift for the $n$-th RE is expressed as $\theta_{k,n} = arg(g_{d,k}) - arg(\bm{b}_{k,n}^H) - arg(\bar{\bm{\varphi}}_{k,n})$. Finally, we obtain that $\theta_{k,n} = arg(g_{d,k}) - arg(\bm{g}_{u,k,n}) +arg(\bm{g}_{r,k,n})$, $\forall k$, $\forall n$. This completes the proof of Proposition \ref{Prooptimal_ris_phase}.

\bibliographystyle{gbt7714-numerical}
\bibliography{citations}

\begin{thebibliography}{41}
\providecommand{\natexlab}[1]{#1}
\providecommand{\url}[1]{#1}
\expandafter\ifx\csname urlstyle\endcsname\relax\else
  \urlstyle{same}\fi
\expandafter\ifx\csname href\endcsname\relax
  \DeclareUrlCommand\doi{\urlstyle{rm}}
  \def\eprint#1#2{#2}
\else
  \def\doi#1{\href{https://doi.org/#1}{\nolinkurl{#1}}}
  \let\eprint\href
\fi

\bibitem[Somov et~al.(2015)Somov and Giaffreda]{somov2015Powering}
SOMOV A, GIAFFREDA R.
\newblock Powering {{IoT}} devices: {{Technologies}} and
  opportunities\allowbreak[J].
\newblock IEEE IoT Newsletter, 2015.

\bibitem[Varshney(2008)]{varshney2008Transporting}
VARSHNEY L~R.
\newblock Transporting information and energy
  simultaneously\allowbreak[C]//\allowbreak
2008 {{IEEE}} International Symposium on Information Theory.
\newblock {IEEE}, 2008: 1612-1616.

\bibitem[Zhong et~al.(2014)Zhong, Suraweera, Zheng, Krikidis, and
  Zhang]{zhong2014Wireless}
ZHONG C, SURAWEERA H~A, ZHENG G, et~al.
\newblock Wireless information and power transfer with full duplex
  relaying\allowbreak[J].
\newblock IEEE Transactions on Communications, 2014, 62\allowbreak (10):
  3447-3461.

\bibitem[Wu et~al.(2021)Wu, Guan, and Zhang]{wu2021Intelligent}
WU Q, GUAN X, ZHANG R.
\newblock Intelligent reflecting surface-aided wireless energy and information
  transmission: {{An}} overview\allowbreak[J].
\newblock Proceedings of the IEEE, 2021.

\bibitem[Zhang et~al.(2013)Zhang and Ho]{zhang2013MIMO}
ZHANG R, HO C~K.
\newblock {{MIMO}} broadcasting for simultaneous wireless information and power
  transfer\allowbreak[J].
\newblock IEEE Transactions on Wireless Communications, 2013, 12\allowbreak
  (5): 1989-2001.

\bibitem[Ju et~al.(2013)Ju and Zhang]{ju2013Throughput}
JU H, ZHANG R.
\newblock Throughput maximization in wireless powered communication
  networks\allowbreak[J].
\newblock IEEE Transactions on Wireless Communications, 2013, 13\allowbreak
  (1): 418-428.

\bibitem[Ju et~al.(2014)Ju and Zhang]{ju2014user}
JU H, ZHANG R.
\newblock User cooperation in wireless powered communication
  networks\allowbreak[C]//\allowbreak
2014 {{IEEE Global Communications Conference}}.
\newblock {IEEE}, 2014: 1430-1435.

\bibitem[Ju et~al.(2014)Ju and Zhang]{ju2014Optimal}
JU H, ZHANG R.
\newblock Optimal resource allocation in full-duplex wireless-powered
  communication network\allowbreak[J].
\newblock IEEE Transactions on Communications, 2014, 62\allowbreak (10):
  3528-3540.

\bibitem[Kim et~al.(2016)Kim, Lee, Song, Oh, and Lee]{kim2016Sum}
KIM J, LEE H, SONG C, et~al.
\newblock Sum throughput maximization for multi-user {{MIMO}} cognitive
  wireless powered communication networks\allowbreak[J].
\newblock IEEE Transactions on Wireless Communications, 2016, 16\allowbreak
  (2): 913-923.

\bibitem[Wu et~al.(2015)Wu, Tao, Ng, Chen, and Schober]{wu2015Energyefficient}
WU Q, TAO M, NG D~W~K, et~al.
\newblock Energy-efficient resource allocation for wireless powered
  communication networks\allowbreak[J].
\newblock IEEE Transactions on Wireless Communications, 2015, 15\allowbreak
  (3): 2312-2327.

\bibitem[Wu et~al.(2015)Wu, Chen, and Li]{wu2015Wirelessa}
WU Q, CHEN W, LI J.
\newblock Wireless powered communications with initial energy: {{QoS}}
  guaranteed energy-efficient resource allocation\allowbreak[J].
\newblock IEEE Communications Letters, 2015, 19\allowbreak (12): 2278-2281.

\bibitem[Wu et~al.(2019)Wu and Zhang]{wu2019Intelligent}
WU Q, ZHANG R.
\newblock Intelligent {{Reflecting Surface Enhanced Wireless Network}} via
  {{Joint Active}} and {{Passive Beamforming}}\allowbreak[J].
\newblock IEEE Transactions on Wireless Communications, 2019, 18\allowbreak
  (11): 5394-5409.

\bibitem[Di~Renzo et~al.(2020)Di~Renzo, Zappone, Debbah, Alouini, Yuen,
  De~Rosny, and Tretyakov]{direnzo2020Smart}
DI~RENZO M, ZAPPONE A, DEBBAH M, et~al.
\newblock Smart radio environments empowered by reconfigurable intelligent
  surfaces: {{How}} it works, state of research, and the road
  ahead\allowbreak[J].
\newblock IEEE journal on selected areas in communications, 2020, 38\allowbreak
  (11): 2450-2525.

\bibitem[Kundu et~al.(2020)Kundu and McKay]{kundu2020RISAssisted}
KUNDU N~K, MCKAY M~R.
\newblock {{RIS-Assisted MISO Communication}}: {{Optimal Beamformers}} and
  {{Performance Analysis}}\allowbreak[C]//\allowbreak
2020 {{IEEE Globecom Workshops}} ({{GC Wkshps}}.
\newblock {Taipei, Taiwan}: {IEEE}, 2020: 1-6.

\bibitem[Zou et~al.(2020)Zou, Gong, Xu, Cheng, Hoang, and
  Niyato]{zou2020Wireless}
ZOU Y, GONG S, XU J, et~al.
\newblock Wireless powered intelligent reflecting surfaces for enhancing
  wireless communications\allowbreak[J].
\newblock IEEE Transactions on Vehicular Technology, 2020, 69\allowbreak (10):
  12369-12373.

\bibitem[Zhang et~al.(2021)Zhang and Dai]{zhang2021Joint}
ZHANG Z, DAI L.
\newblock A joint precoding framework for wideband reconfigurable intelligent
  surface-aided cell-free network\allowbreak[J].
\newblock IEEE Transactions on Signal Processing, 2021, 69: 4085-4101.

\bibitem[Liang et~al.(2021)Liang, Chen, Long, He, Lin, Huang, Liu, Shen, and
  Di~Renzo]{liang2021Reconfigurablea}
LIANG Y~C, CHEN J, LONG R, et~al.
\newblock Reconfigurable intelligent surfaces for smart wireless environments:
  Channel estimation, system design and applications in {{6G}}
  networks\allowbreak[J].
\newblock Science China Information Sciences, 2021, 64\allowbreak (10): 1-21.

\bibitem[Yu et~al.(2021)Yu, Jamali, Xu, Ng, and Schober]{yu2021Smart}
YU X, JAMALI V, XU D, et~al.
\newblock Smart and reconfigurable wireless communications: {{From IRS}}
  modeling to algorithm design\allowbreak[J].
\newblock IEEE Wireless Communications, 2021, 28\allowbreak (6): 118-125.

\bibitem[Zhang et~al.(2021)Zhang, Dai, Chen, Liu, Yang, Schober, and
  Poor]{zhang2021Active}
ZHANG Z, DAI L, CHEN X, et~al.
\newblock Active {{RIS}} vs. passive {{RIS}}: {{Which}} will prevail in
  {{6G}}?\allowbreak[A].
\newblock 2021.
\newblock arXiv: \eprint{https://arxiv.org/abs/2103.15154}{2103.15154}.

\bibitem[Long et~al.(2021)Long, Liang, Pei, and Larsson]{long2021Active}
LONG R, LIANG Y~C, PEI Y, et~al.
\newblock Active reconfigurable intelligent surface-aided wireless
  communications\allowbreak[J].
\newblock IEEE Transactions on Wireless Communications, 2021, 20\allowbreak
  (8): 4962-4975.

\bibitem[Lyu et~al.(2021)Lyu, Ramezani, Hoang, Gong, Yang, and
  Jamalipour]{lyu2021Optimized}
LYU B, RAMEZANI P, HOANG D~T, et~al.
\newblock Optimized {{Energy}} and {{Information Relaying}} in
  {{Self-Sustainable IRS-Empowered WPCN}}\allowbreak[J].
\newblock IEEE TRANSACTIONS ON COMMUNICATIONS, 2021, 69\allowbreak (1): 15.

\bibitem[Zheng et~al.(2021)Zheng, Bi, Zhang, Lin, and Wang]{zheng2021Joint}
ZHENG Y, BI S, ZHANG Y~J~A, et~al.
\newblock Joint beamforming and power control for throughput maximization in
  {{IRS-assisted MISO WPCNs}}\allowbreak[J].
\newblock IEEE Internet of Things Journal, 2021, 8\allowbreak (10): 8399-8410.

\bibitem[Hua et~al.(2022)Hua, Wu, and Poor]{hua2022PowerEfficient}
HUA M, WU Q, POOR H~V.
\newblock Power-{{Efficient Passive Beamforming}} and {{Resource Allocation}}
  for {{IRS-Aided WPCNs}}\allowbreak[J].
\newblock IEEE Transactions on Communications, 2022.

\bibitem[Xu et~al.(2021)Xu, Gao, Wang, Huang, Yang, and
  Yuen]{xu2021RISenhanced}
XU Y, GAO Z, WANG Z, et~al.
\newblock {{RIS-enhanced WPCNs}}: {{Joint}} radio resource allocation and
  passive beamforming optimization\allowbreak[J].
\newblock IEEE Transactions on Vehicular Technology, 2021, 70\allowbreak (8):
  7980-7991.

\bibitem[You et~al.(2021)You and Zhang]{you2021Wireless}
YOU C, ZHANG R.
\newblock Wireless communication aided by intelligent reflecting surface:
  {{Active}} or passive?\allowbreak[J].
\newblock IEEE Wireless Communications Letters, 2021, 10\allowbreak (12):
  2659-2663.

\bibitem[Dong et~al.(2021)Dong, Wang, and Bai]{dong2021Active}
DONG L, WANG H~M, BAI J.
\newblock Active {{Reconfigurable Intelligent Surface Aided Secure
  Transmission}}\allowbreak[J].
\newblock IEEE Transactions on Vehicular Technology, 2021.

\bibitem[Zargari et~al.(2022)Zargari, Hakimi, Tellambura, and
  Herath]{zargari2022Multiuser}
ZARGARI S, HAKIMI A, TELLAMBURA C, et~al.
\newblock Multiuser {{MISO PS-SWIPT Systems}}: {{Active}} or {{Passive
  RIS}}?\allowbreak[J].
\newblock IEEE Wireless Communications Letters, 2022.

\bibitem[Gao et~al.(2022)Gao, Wu, Zhang, Chen, Ng, and
  Di~Renzo]{gao2022Beamforming}
GAO Y, WU Q, ZHANG G, et~al.
\newblock Beamforming {{Optimization}} for {{Active Intelligent Reflecting
  Surface-Aided SWIPT}}\allowbreak[A].
\newblock 2022.
\newblock arXiv: \eprint{https://arxiv.org/abs/2203.16093}{2203.16093}.

\bibitem[Zeng et~al.(2022)Zeng, Qiao, Wu, and Wu]{zeng2022Throughput}
ZENG P, QIAO D, WU Q, et~al.
\newblock Throughput {{Maximization}} for {{Active Intelligent Reflecting
  Surface Aided Wireless Powered Communications}}\allowbreak[J].
\newblock IEEE Wireless Communications Letters, 2022.

\bibitem[Wu et~al.(2020)Wu and Zhang]{wu2020Weighted}
WU Q, ZHANG R.
\newblock Weighted {{Sum Power Maximization}} for {{Intelligent Reflecting
  Surface Aided SWIPT}}\allowbreak[J].
\newblock IEEE WIRELESS COMMUNICATIONS LETTERS, 2020, 9\allowbreak (5): 5.

\bibitem[Hu et~al.(2021)Hu, Wei, Cai, Liu, Ng, and Yuan]{hu2021robust}
HU S, WEI Z, CAI Y, et~al.
\newblock Robust and secure sum-rate maximization for multiuser {{MISO}}
  downlink systems with self-sustainable {{IRS}}\allowbreak[J].
\newblock IEEE Transactions on Communications, 2021, 69\allowbreak (10):
  7032-7049.

\bibitem[Zheng et~al.(2019)Zheng and Zhang]{zheng2019Intelligent}
ZHENG B, ZHANG R.
\newblock Intelligent reflecting surface-enhanced {{OFDM}}: {{Channel}}
  estimation and reflection optimization\allowbreak[J].
\newblock IEEE Wireless Communications Letters, 2019, 9\allowbreak (4):
  518-522.

\bibitem[Wu et~al.(2021)Wu, Zhang, Zheng, You, and Zhang]{wu2021Intelligenta}
WU Q, ZHANG S, ZHENG B, et~al.
\newblock Intelligent reflecting surface-aided wireless communications: {{A}}
  tutorial\allowbreak[J].
\newblock IEEE Transactions on Communications, 2021, 69\allowbreak (5):
  3313-3351.

\bibitem[Wang et~al.(2020)Wang, Liu, and Cui]{wang2020Channel}
WANG Z, LIU L, CUI S.
\newblock Channel {{Estimation}} for {{Intelligent Reflecting Surface Assisted
  Multiuser Communications}}: {{Framework}}, {{Algorithms}}, and
  {{Analysis}}\allowbreak[J].
\newblock IEEE Transactions on Wireless Communications, 2020, 19\allowbreak
  (10): 6607-6620.

\bibitem[Zheng et~al.(2020)Zheng, Bi, Zhang, Quan, and
  Wang]{zheng2020Intelligent}
ZHENG Y, BI S, ZHANG Y~J, et~al.
\newblock Intelligent reflecting surface enhanced user cooperation in wireless
  powered communication networks\allowbreak[J].
\newblock IEEE Wireless Communications Letters, 2020, 9\allowbreak (6):
  901-905.

\bibitem[Luo et~al.(2010)Luo, Ma, So, Ye, and Zhang]{luo2010Semidefinite}
LUO Z~Q, MA W~K, SO A, et~al.
\newblock Semidefinite {{Relaxation}} of {{Quadratic Optimization
  Problems}}\allowbreak[J].
\newblock IEEE Signal Processing Magazine, 2010, 27\allowbreak (3): 20-34.

\bibitem[Boyd et~al.(2004)Boyd, Boyd, and Vandenberghe]{boyd2004Convex}
BOYD S, BOYD S~P, VANDENBERGHE L.
\newblock Convex optimization\allowbreak[M].
\newblock {Cambridge university press}, 2004.

\bibitem[Razaviyayn(2014)]{razaviyayn2014Successive}
RAZAVIYAYN M.
\newblock Successive {{Convex Approximation}}: {{Analysis}} and
  {{Applications}}\allowbreak[D].
\newblock University of Minnesota, 2014.

\bibitem[Guo et~al.(2020)Guo, Liang, Chen, and Larsson]{guo2020Weighted}
GUO H, LIANG Y~C, CHEN J, et~al.
\newblock Weighted {{Sum-Rate Maximization}} for {{Reconfigurable Intelligent
  Surface Aided Wireless Networks}}\allowbreak[J].
\newblock IEEE Transactions on Wireless Communications, 2020, 19\allowbreak
  (5): 3064-3076.

\bibitem[Amato et~al.(2018)Amato, Peterson, Degnan, and
  Durgin]{amato2018Tunneling}
AMATO F, PETERSON C~W, DEGNAN B~P, et~al.
\newblock Tunneling {{RFID}} tags for long-range and low-power microwave
  applications\allowbreak[J].
\newblock IEEE Journal of Radio Frequency Identification, 2018, 2\allowbreak
  (2): 93-103.

\bibitem[Xu et~al.(2021)Xu, Liang, Yang, and Zhao]{xu2021Reconfigurable}
XU X, LIANG Y~C, YANG G, et~al.
\newblock Reconfigurable intelligent surface empowered symbiotic radio over
  broadcasting signals\allowbreak[J].
\newblock IEEE Transactions on Communications, 2021, 69\allowbreak (10):
  7003-7016.

\end{thebibliography}

\end{document}